\documentclass[11pt]{amsart}
\usepackage[margin=2.5cm]{geometry}
\usepackage{enumerate}
\usepackage{amsmath}
\usepackage{amssymb,latexsym}
\usepackage{amsthm}
\usepackage{color}
\usepackage[dvipsnames]{xcolor}
\usepackage{cancel}
\usepackage{graphicx}
\usepackage{cite}
\usepackage{changes}
\makeatletter
\renewcommand{\tocsection}[3]{%
  \indentlabel{\@ifnotempty{#2}{\bfseries\ignorespaces#1 #2\quad}}\bfseries#3}
\renewcommand{\tocsubsection}[3]{%
  \indentlabel{\@ifnotempty{#2}{\ignorespaces#1 #2\quad}}#3}

\newcommand\@dotsep{4.5}
\def\@tocline#1#2#3#4#5#6#7{\relax
  \ifnum #1>\c@tocdepth 
  \else
    \par \addpenalty\@secpenalty\addvspace{#2}%
    \begingroup \hyphenpenalty\@M
    \@ifempty{#4}{%
      \@tempdima\csname r@tocindent\number#1\endcsname\relax
    }{%
      \@tempdima#4\relax
    }%
    \parindent\z@ \leftskip#3\relax \advance\leftskip\@tempdima\relax
    \rightskip\@pnumwidth plus1em \parfillskip-\@pnumwidth
    #5\leavevmode\hskip-\@tempdima{#6}\nobreak
    \leaders\hbox{$\m@th\mkern \@dotsep mu\hbox{.}\mkern \@dotsep mu$}\hfill
    \nobreak
    \hbox to\@pnumwidth{\@tocpagenum{\ifnum#1=1\bfseries\fi#7}}\par
    \nobreak
    \endgroup
  \fi}
\AtBeginDocument{%
\expandafter\renewcommand\csname r@tocindent0\endcsname{0pt}
}
\def\l@subsection{\@tocline{2}{0pt}{2.5pc}{5pc}{}}
\makeatother

\usepackage{etoolbox}
\makeatletter
\patchcmd{\@setaddresses}{\indent}{\noindent}{}{}
\patchcmd{\@setaddresses}{\indent}{\noindent}{}{}
\patchcmd{\@setaddresses}{\indent}{\noindent}{}{}
\patchcmd{\@setaddresses}{\indent}{\noindent}{}{}
\makeatother

\usepackage[linktocpage]{hyperref}
\hypersetup{
  colorlinks   = true, 
  urlcolor     = blue, 
  linkcolor    = Purple, 
  citecolor   = red 
}
\usepackage{comment}
\usepackage{amscd}
\usepackage{mathtools}

\DeclareMathOperator{\C}{\mathcal{C}}
\newcommand{\srk}{\mathrm{srk}}
\newcommand{\dsrk}{\mathrm{d}_{\mathrm{srk}}}
\DeclareMathOperator{\GammaL}{\Gamma\mathrm{L}}

\DeclareMathOperator{\Gal}{Gal}
\DeclareMathOperator{\supp}{supp}
\DeclareMathOperator{\rk}{rk}
\DeclareMathOperator{\HH}{H}
\DeclareMathOperator{\mm}{m}
\DeclareMathOperator{\dd}{d}
\DeclareMathOperator{\Mat}{Mat}
\DeclareMathOperator{\Ext}{Ext}
\DeclareMathOperator{\colsp}{colsp}
\DeclareMathOperator{\PGL}{PGL}
\DeclareMathOperator{\GL}{GL}

\DeclareMathOperator{\ww}{w}

\theoremstyle{definition}
\newtheorem{theorem}{Theorem}[section]

\newtheorem{lemma}[theorem]{Lemma}
\newtheorem{corollary}[theorem]{Corollary}
\newtheorem{definition}[theorem]{Definition}
\newtheorem{proposition}[theorem]{Proposition}

\newtheorem{example}[theorem]{Example}

\newtheorem{remark}[theorem]{Remark}

\newtheorem{open}[theorem]{Problem}

\newcommand{\cC}{{\mathcal C}}

\newcommand{\cM}{{\mathcal M}}
\newcommand{\cG}{{\mathcal G}}

\newcommand{\cH}{{\mathcal H}}

\newcommand{\cV}{{\mathcal V}}
\newcommand{\cO}{{\mathcal O}}
\newcommand{\cU}{{\mathcal U}}
\newcommand{\mU}{{\mathcal U}}
\newcommand{\mV}{{\mathcal V}}

\newcommand{\F}{{\mathbb F}}
\newcommand{\D}{{\mathbb D}}
\newcommand{\Z}{{\mathbb Z}}
\newcommand{\NN}{{\mathbb N}}

\newcommand{\bfn}{\mathbf {n}}
\newcommand{\bfm}{\mathbf {m}}

\newcommand{\fq}{{\mathbb F}_{q}}
\newcommand{\Fq}{{\mathbb F}_{q}}
\newcommand{\Fm}{{\mathbb F}_{q^m}}

\newcommand{\la}{\langle}
\newcommand{\ra}{\rangle}

\newcommand{\PG}{\mathrm{PG}}
\newcommand{\N}{\mathrm{N}}

\newcommand{\spac}{\Mat(\bfn,\bfm,\Fq)}
\newcommand{\Fmnk}{[\bfn,k]_{q^m/q}}
\newcommand{\Fmnkd}{[\bfn,k,d]_{q^m/q}}
\newcommand{\Fmk}{[n,k]_{q^m/q}}

\newcommand{\Fmkt}{(n,k)_{q^m/q}}

\newcommand{\st}{\,:\,}

\title{The geometry of one-weight codes in the sum-rank metric}
\date{}
\author[A. Neri]{Alessandro Neri}
\address{Alessandro Neri, \textnormal{Max-Planck-Institute for Mathematics in the Sciences, Inselstraße 22, 04103 Leipzig, Germany}}
\email{alessandro.neri@mis.mpg.de}

\author[P. Santonastaso]{Paolo Santonastaso}
\address{Paolo Santonastaso, \textnormal{Dipartimento di Matematica e Fisica, Universit\`a degli Studi della Campania ``Luigi Vanvitelli'', Viale Lincoln, 5, I--\,81100 Caserta, Italy}}
\email{paolo.santonastaso@unicampania.it}

\author[F. Zullo]{Ferdinando Zullo}
\address{Ferdinando Zullo, \textnormal{Dipartimento di Matematica e Fisica, Universit\`a degli Studi della Campania ``Luigi Vanvitelli'', Viale Lincoln, 5, I--\,81100 Caserta, Italy}}
\email{ferdinando.zullo@unicampania.it}

\begin{document}

\maketitle

\begin{abstract}
We provide a geometric characterization of $k$-dimensional $\Fm$-linear sum-rank metric codes as tuples of $\Fq$-subspaces of $\Fm^k$. We then use this characterization to study one-weight codes in the sum-rank metric. This leads us to extend the family of linearized Reed-Solomon codes in order to obtain a doubly-extended version of them. We prove that these codes are still maximum sum-rank distance (MSRD) codes and, when $k=2$, they are one-weight, as in the Hamming-metric case.
We then focus on constant rank-profile codes in the sum-rank metric, which are a special family of one weight-codes, and derive constraints on their parameters with the aid of an associated Hamming-metric code.
Furthermore, we introduce the $n$-simplex codes in the sum-rank metric, which are obtained as the orbit of a Singer subgroup  of $\GL(k,q^m)$. They turn out to be constant rank-profile -- and hence one-weight -- and generalize the simplex codes in both the Hamming and the rank metric.
Finally, we focus on $2$-dimensional one-weight codes, deriving constraints on the parameters of those which are also MSRD, and we find a new construction of one-weight MSRD codes when $q=2$.
\end{abstract}

\tableofcontents

\section{Introduction}

Codes endowed with the sum-rank metric have been demonstrated to be useful in several contexts. Their notion can be traced back to \cite{el2003design,lu2005unified}, where they have been linked to space-time coding. However, they became popular when they were used for improving the performance of multishot network coding built on rank-metric codes \cite{nobrega2010multishot}. Later, sum-rank-metric codes have also been used in distributed storage, 
playing a crucial role in some special constructions of partial MDS codes \cite{martinez2019universal,martinez2020general}.

In the last few years, a deep mathematical theory of sum-rank-metric codes was developed in a series of papers by Martínez-Peñas \cite{martinezpenas2018skew,martinez2019theory,martinezpenas2021hamming}. These codes can be seen as a generalization of Hamming-metric codes and rank-metric codes. Vectors are first partitioned in blocks of (possibly) variable lengths, then on each block one considers the rank-metric and  sums up the distances of each block. Considering only one block corresponds to the rank metric setting, while choosing all the blocks to have length one is equivalent to equip the ambient space with the Hamming metric. 

On the other hand, the theory of error-correcting codes has often been influenced by finite geometry, and the interplay between these two subjects has emerged since the late 50's. Among their connections, it is well-known that linear codes in the Hamming metric correspond to \emph{projective systems}, that are sets of points in a finite projective space. This correspondence translates metric properties of a linear code into geometric properties of the associated projective system; see \cite[Section 1.1]{tsfasman1991algebraic}. This geometric viewpoint is very useful for the classification of one-weight codes shown by Bonisoli in \cite{bonisoli1983every}. 
More recently a similar geometric approach was used for codes in the rank metric. A first intuition was given in \cite{sheekey2016new}, while a full correspondence was provided in \cite{sheekey2019scatterd} and in \cite{Randrianarisoa2020ageometric}: $k$-dimensional rank-metric codes over an extension field $\Fm$ of $\Fq$ correspond to \emph{$q$-systems}, which are special $\Fq$-subspaces of $\Fm^k$. This correspondence allowed us to give a complete characterization of nondegenerate $\Fm$-linear one-weight codes in the rank metric; see \cite[Theorem 12]{Randrianarisoa2020ageometric}, \cite[Proposition 3.16]{alfarano2021linear}. 

\bigskip

\noindent \textbf{Our contribution.} Motivated by these results, in this paper we give a complete geometric characterization of $\Fm$-linear sum-rank metric codes based on Theorem \ref{th:connection} and Theorem \ref{thm:1-1corr}, and then we use this geometric point of view to study one-weight codes in the sum-rank metric. In contrast to what happens for the Hamming and the rank metric, in which there is essentially only one nondegenerate one-weight code, we show that there are many different one-weight codes in the sum-rank metric. 
We first extend linearized Reed-Solomon codes to \emph{doubly-extended linearized Reed-Solomon codes}, which we prove to be maximum sum-rank distance codes in Theorem \ref{th:MSRDLRS}. Furthermore, in Theorem \ref{th:RS1w} we show that the two-dimensional ones are one-weight, as it happens for the Hamming-metric case. 
We then restrict to study constant rank-profile codes, which are sum-rank metric codes such that for every nonzero block-codeword $c=(c_1 \mid \ldots \mid c_t)$, the multiset $\{\{\rk_q(c_1),\ldots,\rk_q(c_t)\}\}$ is always the same. We associate a Hamming-metric code to a sum-rank metric code exploiting the geometric correspondence, and use such a code for deriving constraints on constant rank-profile codes; see Corollary \ref{cor:parameters}.
In Theorem \ref{thm:orbit_construction} we provide an orbital construction of constant rank-profile codes, using the action of a transitive subgroup $\mathcal G$ of $\GL(k,q^m)$. This leads to the notion of \emph{$n$-simplex code} in the sum-rank metric given in Definition \ref{def:simplex}. We observe then that this code specializes to the simplex code in the Hamming metric when $n=1$, and to the simplex code in the rank metric whenever the action of $\mathcal G$ is trivial; see Remark \ref{rem:simplex_special_cases}. Finally, we investigate $2$-dimensional one-weight codes which are also MSRD, introducing the notions of \emph{multi-linear sets} and \emph{scattered multi-linear sets}. First,  we prove that these codes must have at least $q+1$ blocks. Then, in Theorem \ref{th:boundq+1} we show that for $q\geq 3$, whenever $2$-dimensional one-weight MSRD codes with exactly $q+1$ blocks exist, they must have the same block lengths of the doubly-extended linearized Reed-Solomon. For $q=2$ in Theorem \ref{thm:oneweight_MSRD_sporadic} we find a special family of one-weight MSRD codes, which we call \emph{$2$-fold linearized Reed-Solomon codes}, having different block-lengths. We conclude providing additional constructions of $2$-dimensional one-weight codes, using the \emph{lifting construction}. 
\bigskip

\noindent \textbf{Outline.} The paper is structured as follows. 
Section \ref{sec:2}  contains the basic notions on sum-rank metric codes. 
In Section \ref{sec:3} we derive the geometric characterization of sum-rank metric codes in terms of multiple $q$-systems. 
Section \ref{sec:4} is dedicated to the study of doubly-extended linearized Reed-Solomon codes. We then associate a sum-rank metric code with a special Hamming-metric code in Section \ref{sec:extendHamm}, and derive constraints on the parameters of one-weight codes. 
In Section \ref{sec:orbital} we provide an orbital construction of one-weight sum-rank metric codes using transitive groups, and give the counterpart of simplex codes in the sum-rank metric. 
In Section \ref{sec:non-hom} we focus on $2$-dimensional one-weight codes, characterizing the block structure of those that are also MSRD.
Finally,  we draw our conclusions and we describe some natural research questions in Section \ref{sec:8}.

\bigskip

\noindent \textbf{Notation.} We fix now the notation that we will use for the whole paper. For us $q$ is a prime power and $\Fq$ is the finite field with $q$ elements. We will often consider the degree $m$ extension field $\Fm$ of $\Fq$. 
\noindent Let $t$ be a positive integer. From now on
$\bfn=(n_1,\ldots,n_t), \bfm=(m_1,\ldots,m_t) \in \NN^t$ will always be ordered tuples with $n_1 \geq n_2 \geq \ldots \geq n_t$ and $m_1 \geq m_2, \geq \ldots \geq m_t$, and we set $N\coloneqq n_1+\ldots+n_t$. We use the following compact notations for the direct sum of vector spaces 
$$ \Fq^\bfn\coloneqq\bigoplus_{i=1}^t\Fq^{n_i}, \qquad  \Fm^\bfn\coloneqq\bigoplus_{i=1}^t\Fm^{n_i},$$
and for the direct sum of matrix spaces
$$\spac\coloneqq \bigoplus_{i=1}^t \F_q^{n_i \times m_i}.$$
Furthermore, we follow the notation used in \cite[Section 3.3]{alfarano2021sum}. For a vector $\mathbf{a}=(a_1,\ldots, a_r)\in \NN^r$, we define $$\mathcal{S}_{\mathbf{a}}=\mathcal{S}_{a_{1}} \times \ldots \times \mathcal{S}_{a_r},$$
where $\mathcal{S}_i$ is the symmetric group of order $i$. Similarly, we denote by $\GL(\mathbf{a}, \F_q)$ the direct product of the general linear groups of degree $a_i$ over $\F_q$, i.e.
$$ \GL(\mathbf{a}, \F_q) = \GL(a_1, \F_q)\times \ldots\times \GL(a_r, \F_q).$$
Moreover, we denote by $\lambda(\bfn)=(\lambda_1,\ldots,\lambda_s)$ the vector of positive integers which counts the occurrences of the distinct entries of $\bfn$. Formally, let $\mathcal N(\bfn)\coloneqq\{n_1, \ldots, n_t\}=\{n_{i_1},\ldots,n_{i_s}\}$, with $n_{i_1}>\ldots > n_{i_s}$ and $s\coloneqq|\mathcal N (\bfn)|$. Then
$\lambda_j \coloneqq|\{k \st n_k = n_{i_j}\}|$ for each $j\in [s]\coloneqq\{1,\ldots,s\}$.

\noindent Finally, whenever we will talk about duality in this paper, we will consider the standard inner product on the ambient spaces, and denote the \textbf{dual} of a space $V$ with respect to this bilinear form by $V^\perp$.

\section{Sum-rank-metric codes, supports and isometries}\label{sec:2}

\subsection{Matrix codes}
In this section we recall the basic notions of sum-rank-metric codes seen as elements in $\spac$, that will be useful for the rest of the paper. The interested reader is referred to \cite{byrne2021fundamental,byrne2020anticodes,moreno2021optimal} for a more detailed description of this setting. 

\begin{definition}
Let $X\coloneqq(X_1,\dots, X_t)\in\Mat(\mathbf{n},\mathbf{m},\Fq)$. 
The \textbf{sum-rank support} of $X$ is defined as the space
$$\supp(X)\coloneqq(\colsp(X_1), \colsp(X_2),\ldots, \colsp(X_t)) \subseteq \Fq^\bfn,$$
where $\colsp(X_i)$ is the $\F_q$-span of the columns of $X_i$.
The \textbf{rank-list} of $X$ is defined as 
$$\rho(X) \coloneqq (\dim(\colsp(X_1)),\dots, \dim(\colsp(X_t)))\in\mathbb{N}^t.$$
The \textbf{rank-profile} of $X$ is the $t$-uple obtained from $\rho(X)$ by rearranging its  entries in non-increasing order and it is denoted by $\mu(X)$. Finally, the \textbf{sum-rank weight} of $X$ is the quantity
$$\ww_{\srk}(X)\coloneqq\dim_{\Fq}(\supp(X))=\sum_{i=1}^t \rk(X_i).$$
\end{definition}

With these definition in mind, we can endow the space $\spac$ with a distance function, called the \textbf{sum-rank distance},
\[
\dsrk : \spac \times \spac \longrightarrow \mathbb{N}  
\]
defined by
\[
\dsrk(X,Y) \coloneqq \ww_{\srk}(X-Y).
\]

\begin{definition}
 A \textbf{(matrix) sum-rank metric code} $\mathcal{C}$ is an $\F_q$-linear subspace of $\spac$ endowed with the sum-rank distance. 
The \textbf{minimum sum-rank distance} of a sum-rank code $\mathcal{C}$ is defined as usual via $$\dsrk(\mathcal{C})\coloneqq\min\{\ww_{\srk}(X) \st X \in \mathcal{C}, X \neq \mathbf{0}\}.$$ 
The \textbf{sum-rank support} of the code $\mathcal{C}$ is the $\Fq$-span of the supports of all the codewords of $\C$, that is
$$ \supp(\C)\coloneqq\sum_{X\in\C} \supp(X) \subseteq \Fq^\bfn.$$
Finally,  we say that $\mathcal{C}$ is \textbf{sum-rank nondegenerate} if $\supp(\mathcal{C})=\F_q^\bfn$.
\end{definition}

\subsection{Vector codes} \label{sub:vectorsetting}
From now on, we consider the case in which $m=m_1=\cdots=m_t$.
We recall a slightly different notion of sum-rank metric code, in which the codewords are vectors with entries from an extension field $\F_{q^m}$ rather than matrices over $\F_q$. The interested reader is referred to \cite{martinezpenas2018skew,martinez2019universal,martinez2019theory,neri2021twisted,ott2021bounds} for a more detailed description of this setting.

The \textbf{$\Fq$-rank} of a vector $v=(v_1,\ldots,v_n) \in \F_{q^m}^n$ is the $\Fq$-dimension of the
vector space generated over $\F_q$ by its entries, i.e, 
$$\rk_q(v)\coloneqq\dim_{\fq} \langle v_1,\ldots, v_n\rangle_{\fq}.$$ 
Let $x=(x_1 , \ldots,  x_t)\in\F_{q^m}^\bfn$, with $x_i\in\F_{q^m}^{n_i}$ for any $i$. The \textbf{$\Fq$-sum-rank weight of $x$} is defined as the quantity
$$ \ww(x)=\sum_{i=1}^t \rk_q(x_i).$$
Consequently, the 
\textbf{$\Fq$-sum-rank distance on $\Fm^\bfn$}, is  
\[
\dd(x,y)=\sum_{i=1}^t \rk_q(x_i-y_i),
\]
for any  $x=(x_1 , \ldots, x_t), y=(y_1, \ldots, y_t) \in \F_{q^m}^\bfn$, with $x_i,y_i \in \F_{q^m}^{n_i}$.

Observe that $\Fm^\bfn$ is naturally isomorphic to $\Fm^N$ by concatenating the defining vectors. However, in order to define the sum-rank weight and sum-rank distance directly on $\Fm^N$ we would need to specify with respect to which partition of $N$ we are computing them. For this reason we consider $\Fm^\bfn$ as ambient space, so that the partition is already specified. However, in some situation we will use this natural representation of the elements of $\Fm^\bfn$ as vectors of length $N$ over $\Fm$.
Thus,  from now on, we will simply write ``sum-rank weight'' and ``sum-rank distance'' when there is no risk of ambiguity on the base field $\Fq$ and on  the integer partition of $N$.

\begin{definition}
Let $k$ be a positive integer with $1 \le k \le N$. A \textbf{(vector) sum-rank metric code} $\C$ is a $k$-dimensional $\Fm$-subspace of $\Fm^\bfn$ endowed with the sum-rank metric. 
The \textbf{minimum sum-rank distance} of $\C$ is the integer
\[
\dd(\C)=\min\{\dd(x,y) \st x, y \in \C, x\neq y  \}= \min\{\ww(x) \st x \in \C, x\neq 0  \}.
\]
For brevity, from now on we will write that $\C$ is an $\Fmnkd$ code if $k$ is the $\Fm$-dimension of $\C$ and $d$ is its minimum distance, or simply an $\Fmnk$ code if the minimum distance is not relevant/known.
\end{definition}

Since the elements of $\Fm^\bfn$ can also be seen as long vectors in $\Fm^N$, every $\Fmnk$ code can be provided with generator and parity-check matrices. Each of them, is naturally partitioned as $G=(G_1 \mid \ldots \mid G_t)$, where $G_i \in \Fm^{k\times n_i}$ (respectively $H=(H_1 \mid \ldots \mid H_t)$, where $H_i \in \Fm^{(N-k)\times n_i}$).

\medskip 

For vector sum-rank metric codes, a Singleton-like bound holds.
\begin{theorem}[see \textnormal{\cite[Proposition 16]{martinezpenas2018skew}}]  \label{th:Singletonbound}
     Let $\mathcal{C}$ be an $\Fmnkd$ code. Then 
     \[
     d \leq N-k+1.
     \]
\end{theorem}

An $\Fmnkd$ code is called a \textbf{Maximum Sum-Rank Distance code} (or shortly \textbf{MSRD code}) if $d=N-k+1$.

The next result establishes the $\Fm$-linear isometries of the space $\F_{q^m}^\bfn$ endowed with the sum-rank distance, and it has been proved in \cite[Theorem 3.7]{alfarano2021sum}. The case $\bfn=(n,\ldots, n)$ was already proved in \cite[Theorem 2]{martinezpenas2021hamming}.
\begin{theorem}
The group of $\F_{q^m}$-linear isometries of the space $(\F_{q^m}^\bfn,\dd)$  is
$$((\F_{q^m}^\ast)^{t} \times \GL(\bfn, \F_q)) \rtimes \mathcal{S}_{\lambda(\bfn)},$$
which (right)-acts as 
  \begin{equation*} (x_1 , \ldots, x_t)\cdot (\mathbf{a},A_1,\ldots, A_t,\pi) \longmapsto (a_1x_{\pi(1)}A_1 \mid \ldots \mid a_tx_{\pi(t)} A_{t}).\end{equation*}
\end{theorem}
From now on, when the vector $\mathbf{a}=(1,\ldots,1)$, the action of $(\mathbf{a},A,\pi)$ on $x\in\Fm^\bfn$, with $A=(A_1,\ldots, A_t)\in \GL(\bfn,\Fq)$, will be  simply denoted by $x\cdot (A,\pi)$, and if in addition $\pi=\mathrm{id}$, we will simply write 
$x\cdot A$. The same notation is then extended naturally to subsets/subspaces  $\mathcal V\subseteq \Fq^\bfn$, as $\mathcal V \cdot (A,\pi) $ and $\mathcal V \cdot A$.

Since in this setting MacWilliams's extension theorem does not hold, see e.g. \cite{barra2015macwilliams}, we will use the $\Fm$-linear isometries of the whole ambient space in order to define the equivalence of sum-rank metric codes.

\begin{definition}\label{def:equiv_codes} We say that two $\Fmnk$ sum-rank metric codes $\cC_1, \cC_2$ are \textbf{equivalent} if there is an $\Fm$-linear isometry $\phi$, such that $\phi(\cC_1)=\cC_2$. 
\end{definition}
We denote the set of equivalence classes of $\Fmnkd$ sum-rank metric codes by $\mathfrak{C}\Fmnkd$.

\subsection{Supports and degeneracy of sum-rank metric codes}

It is known that  to a vector sum-rank metric code we can associate a matrix sum-rank metric code with the same parameters and metric properties in the following way.

For every $r \in [t]$, let $\Gamma_r=(\gamma_1^{(r)},\ldots,\gamma_m^{(r)})$ be an ordered $\Fq$-basis of $\Fm$, and let $\Gamma\coloneqq(\Gamma_1,\ldots,\Gamma_t)$. Given   $x=(x_1, \ldots ,x_t) \in \Fm^\bfn$, with $x_i \in \Fm^{n_i}$, define the element $\Gamma(x)=(\Gamma_1(x_1), \ldots, \Gamma_t(x_t)) \in \spac$, where $\mathbf{m}=(m,\ldots,m)$ and 
$$x_{r,i} = \sum_{j=1}^m \Gamma_r (x_r)_{ij}\gamma_j^{(r)}, \qquad \mbox{ for all } i \in [n_r].$$
In other words, the $r$-th block of $\Gamma(x)$ is the matrix expansion of the vector $x_r$ with respect to the $\Fq$-basis $\Gamma_r$ of $\Fm$.

The following result is straightforward.

\begin{theorem}\label{thm:isometry_vector_matrix}
    For every tuple ordered $\Fq$-basis $\Gamma=(\Gamma_1,\ldots,\Gamma_t)$ of $\Fm$, the map 
    $$\Gamma: \Fm^\bfn \longrightarrow \spac$$
    is an $\Fq$-linear isometry between the metric spaces $(\Fm^\bfn, \dd)$ and $(\spac,\dsrk)$. 
\end{theorem}

\begin{definition}
The \textbf{matrix sum-rank metric code associated} to an $\Fmnk$ code $\C$
with respect to $\Gamma\coloneqq(\Gamma_1,\ldots,\Gamma_t)$ is
$$\Gamma(\C) \coloneqq \{\Gamma (x) \st x \in \C\} \subseteq \spac.$$
\end{definition}

By Theorem \ref{thm:isometry_vector_matrix}, all the metric properties  of an $\Fmnk$ code $\C$ are preserved via the isometry $\Gamma$, and can be found in the associated matrix code $\Gamma(\C)$. However, one has to be careful, because some of the properties might depend on the choice of the $t$-uple of the $\Fq$-bases $\Gamma$. 
It is well-known that for the rank-metric case the notion of support does not depend on the choice of the $\Fq$-basis $\Lambda$ of $\Fm$; see e.g. \cite[Proposition 2.1]{alfarano2021linear}. Thus, it is not surprising that the same happens for the sum-rank support of an element $x\in \Fm^\bfn$. In other words, we have the following result, whose proof is left to the reader.

\begin{proposition}
 Let $\C$ be an $\Fmnk$ code,  let $\Gamma=(\Gamma_1,\ldots,\Gamma_t),  \Lambda=(\Lambda_1,\ldots,\Lambda_t)$ be two tuples of $\Fq$-bases of $\Fm$ and let $x \in \Fm^\bfn$. Then
 \begin{enumerate}
     \item $\supp(\Gamma(x))=\supp(\Lambda(x))$,
     \item $\rho(\Gamma(x))=\rho(\Lambda(x))$, and
     \item $\mu(\Gamma(x))=\mu(\Lambda(x))$.
 \end{enumerate}
\end{proposition}

Hence, we can give the following definitions of sum-rank support, rank-list and rank-profile for elements in $\Fm^\bfn$. For a deeper understanding of these notions we refer the reader to \cite{martinez2019theory}.

\begin{definition}
The \textbf{sum-rank support} of an element $x=(x_1, \ldots, x_t) \in \Fm^{\bfn}$ is the tuple
$$\supp_{\bfn}(x)\coloneqq\supp(\Gamma(x)),$$
for any (and hence all) choice of $\Gamma=(\Gamma_1,\ldots, \Gamma_t)$, where $\Gamma_i$ is an $\Fq$-basis of $\Fm$ for each $i \in[t]$.
The \textbf{rank-list} of $x=(x_1, \ldots, x_t)$ is the $t$-uple $\rho(x)=(\rk_q(x_1), \ldots,\rk_q(x_t))$. The \textbf{rank-profile} of $x$ is the $t$-uple obtained from $\rho(x)$ by rearranging the entries in non-increasing order and it is denoted by $\mu(x)$.
\end{definition}

In the same way, we can extend the notions of sum-rank support and degeneracy for an $\Fmk$ code.
\begin{definition}
 We define the \textbf{sum-rank support of an $\Fmnk$ code} $\C$ as the $\F_q$-span of the supports of the codewords of $\C$ and we denote it by $\supp(\C)$, i.e.
$$ \supp_{\bfn}(\C)\coloneqq\sum_{c\in\C}\supp_{\bfn}(c).$$
Furthermore, we say  that $\C$ is \textbf{sum-rank nondegenerate} if $\supp(\C)=\Fq^\bfn$. We say that $\C$ is \textbf{sum-rank degenerate} if it is not sum-rank nondegenerate.
\end{definition}

\begin{proposition}\label{prop:support_isometry}
 Let $x=(x_1,\ldots,x_t)\in \Fm^\bfn$, and let $A=(A_1,\ldots, A_t)\in \GL(\bfn,\Fq)$. Then
  $$\supp_\bfn(x\cdot A)=\supp_\bfn(x)\cdot A.$$
  In particular, if $\C_1$ and $\C_2$ are two equivalent $\Fmnk$ codes such that $$\supp_{\bfn}(\C_1)=(\mathcal W_1,\ldots \mathcal W_t), \qquad \supp_{\bfn}(\C_2)=(\mathcal Z_1,\ldots \mathcal Z_t),$$ 
  then there exists $\pi \in \mathcal S_{\lambda(\bfn)}$ and $A \in \GL(\bfn, \Fq)$ such that 
  $Z_{i}= A_i^\top \cdot \mathcal W_{\pi(i)},$ for each $i \in [t]$. In other words, we have
   $\supp_\bfn(\C_2)=\supp_\bfn(\C_1) \cdot (A,\pi)$.
\end{proposition}

\begin{proposition}\label{prop:charact_nondegenarate}
Let $\C$ be a $\Fmnk$ code.  The following are equivalent.
\begin{enumerate}
    \item $\C$ is sum-rank nondegenerate.
    \item For every $A=(A_1, \ldots,A_t) \in \GL(\bfn,q)$ the code $\C\cdot A$ is Hamming-nondegenerate.
    \item For any generator matrix $G=(G_1 \,|\, \ldots \,|\, G_t)$ with $G_i \in \Fm^{k \times n_i}$ for $i\in[t]$, the columns of each $G_i$ are $\Fq$-linear independent.
    \item $\dd(\C^\perp)>1$.
\end{enumerate}
\end{proposition}

\begin{proof}
\underline{$(1) \Rightarrow (2)$:} Assume  that there exists $A=(A_1, \ldots,A_t) \in \GL(\bfn,q)$ the code $\C\cdot A$ is Hamming-degenerate. This is equivalent to say that one entry (say the last one) of $\C\cdot A$ is identically zero. This implies that 
$$\supp(\cC\cdot A)\subseteq \F_q^{n_1} \oplus \ldots \oplus  \F_q^{n_{t-1}}\oplus \cV,$$
where $\cV = \{v \in \F_q^{n_t} \st v_{n_t}=0\}$. 
Hence, by Proposition \ref{prop:support_isometry} we have 
$$\supp(\cC) \subseteq \F_q^{n_1} \oplus \ldots \oplus  \F_q^{n_{t-1}}\oplus (A_t^\top)^{-1}\cdot \cV\subsetneq \Fq^\bfn,$$
meaning that $\C$ is sum-rank degenerate.

\underline{$(2) \Rightarrow (3)$:} Suppose that there exists an $i\in [t]$ such that the columns of $G_i$ are $\Fq$-linearly dependent. Then, there exists $A_i\in \GL(n_i,\Fq)$ such that $G_i A_i$ has a zero column. Thus, for any other choices of $A_r \in \GL(n_r, \Fq)$ for $r\neq i$, we have that the code $\C \cdot A$ is Hamming-degenerate.

\underline{$(3) \Rightarrow (4)$:} Suppose that there exists a codeword $c=(c_1,\ldots,c_t) \in \C^\perp$ such that $\ww(c)=1$. This means that there exists $i \in [t]$, $\alpha \in \Fm^*$ and $(\lambda_1,\ldots,\lambda_{n_i})\in \Fq^{n_i}\setminus\{0\}$ such that  $c_j=0$ for all $j \neq i$ and $c_i=\alpha(\lambda_1,\ldots,\lambda_{n_i})$. Hence, $(\lambda_1,\ldots,\lambda_{n_i})$ provides a nontrivial $\Fq$-linear dependence among the columns of $G_i$, for any $G=(G_1 \,|\, \ldots \,|\, G_t)$ generator matrix for $\C$.

\underline{$(4) \Rightarrow (1)$:} Assume by contradiction that $\C$ is sum-rank degenerate. Then, up to sum-rank isometry, we can say that there exists $i\in [t]$ such that 
 $$\supp(\cC)\subseteq \F_q^{n_1} \oplus \ldots \oplus \cV \oplus \ldots \oplus \F_q^{n_t},$$
 where $\cV \subseteq \{v \in \F_q^{n_i} \st v_{n_i}=0\}$. This implies that $(0, \ldots, e_{n_i}, \ldots 0)\in \C^\perp$, where $e_{n_i}$ denotes the $n_i$-th standard basis vector. Hence $\dd(\C^\perp)=1$.
\end{proof}

\section{Geometry of sum-rank metric codes}\label{sec:3}

It is well-known that equivalence classes of Hamming-metric nondegenerate codes are in one-to-one correspondence with equivalence classes of projective systems; see \cite[Theorem 1.1.6]{tsfasman1991algebraic}.  This connection has  been intensively used to get classification results and intriguing constructions in both the areas of coding theory and finite geometry. Recently, in \cite{sheekey2019scatterd,Randrianarisoa2020ageometric} it has been shown that equivalence classes of nondegenerate rank-metric codes are in one-to-one correspondence with equivalence classes of $q$-systems, where the latter constitute the $q$-analogue of projective systems; see also \cite{alfarano2021linear}. At this point, it is natural to ask whether it is possible to construct geometric objects able to capture the structure of sum-rank metric codes which generalize both projective systems and $q$-systems. In the following, we focus on this question.

\medskip

Let $G=(G_1 \,|\, \ldots \,|\, G_t)$, where $G_i \in \Fm^{k\times n_i}$. For each $i\in[t]$, define $\mathcal{U}_i$ to be the $\Fq$-span of the columns of $G_i$ 
and  the maps
\[
\begin{array}{rccl}
\psi_{G_i}:& \fq^{n_i}& \longrightarrow &\mathcal{U}_i \\
&\lambda & \longmapsto & \lambda G_i^\top,\end{array}\]
\[
\begin{array}{rccl}
\psi_{G}:& \fq^{\bfn}& \longrightarrow & \mathcal{U} \\
&(\lambda_1,\ldots,\lambda_t) & \longmapsto & (\psi_{G_1}(\lambda_1),\ldots, \psi_{G_t}(\lambda_t)).\end{array}\]

\begin{theorem}\label{th:connection}
Let $\C$ be a non-degenerate $\Fmnkd$ code with generator matrix $G=(G_1|\ldots|G_t)$.
Let $\mathcal{U}_i$ be the $\F_q$-span of the columns of $G_i$, for $i\in [t]$. Then, for every $v\in \Fm^k$ and $i \in [t]$ we have
$$\psi_{G_i}^{-1}(\mathcal{U}_i \cap v^\perp)=\supp_{n_i}(vG_i)^\perp.$$
In particular, 
$$\supp_\bfn(vG)^\perp=\psi_G^{-1}((\mathcal{U}_1 \cap v^\perp),\ldots,(\mathcal{U}_t \cap v^\perp))$$
and 
\begin{equation}\label{eq:weight_dimension}
\ww(v G) = N - \sum_{i=1}^t \dim_{\fq}(\mathcal{U}_i \cap v^{\perp}).\end{equation} 
\end{theorem}

\begin{proof}
By assumption, $\C$ is non-degenerate and hence by Proposition  \ref{prop:charact_nondegenarate} we have $\dim_{\fq} \mathcal{U}_i=n_i$ for each $i \in [t]$.
Thus, $\psi_{G_i}$ is an  $\Fq$-isomorphism.

Let us now prove the first part of the statement for $i=1$ and let us write $\mathcal{U}=\mathcal{U}_1, G=G_1$ and $n=n_1$.
Let $r\coloneqq n-\dim_{\Fq}(\mathcal{U}\cap v^\perp)$ and let $(g_{r+1}',\ldots, g_{n}')$ be an $\Fq$-basis of $\mathcal{U}\cap v^\perp$ written as column vectors. We can complete it to an $\Fq$-basis $(g_1',\ldots,g_n')$ of $\mathcal{U}$. Let $B\in \GL(n,\Fq)$ be such that $GB=(g_1' \mid \ldots \mid g_n')$. With this choice, we have that
$$ vGB=(x_1,\ldots, x_r, 0, \ldots, 0)$$
and the $x_i$'s must be $\Fq$-linearly independent. Now, we can complete $x_1,\ldots,x_r$ to an $\Fq$-basis $\Gamma\coloneqq (x_1,\ldots, x_m)$ of $\Fm$. It is clear that 
$$\Gamma(vGB)=\begin{pmatrix} I_r & 0 \\
0 & 0
\end{pmatrix}
$$
and hence $\supp_n(vGB)=\langle e_1,\ldots,e_r\rangle_{\Fq}\subseteq \Fq^n$.
Furthermore, by Proposition \ref{prop:support_isometry}, we have that $\supp_n(vG)=\langle e_1,\ldots,e_r\rangle_{\Fq}\cdot B^{-1}$.

On the other hand, by the choice of $B$, we have that $\psi_{GB}^{-1}(\mathcal{U} \cap v^\perp)=\langle e_{r+1},\ldots,e_n\rangle_{\fq}$. Now, for every $\lambda \in \Fq^n$ we have that 
$\psi_{GB}(\lambda)=\lambda(GB)^\top=(\lambda B^\top)G^\top=\psi_{G}(\lambda B^\top)$, and hence $\psi_{GB}^{-1}(y)=\psi_{G}^{-1}(y)(B^\top)^{-1}$ for every $y\in \mathcal{U}$. Thus
$\psi_{G}^{-1}(\mathcal{U} \cap v^\perp)=\psi_{GB}^{-1}(\mathcal{U} \cap v^\perp)\cdot B^\top =\langle e_{r+1},\ldots,e_n \rangle_{\Fq}\cdot B^\top$, proving the first part of the statement.
The second part of the statement immediately follows.
\end{proof}

\begin{remark}
 Note that when we specialize to the Hamming metric, the result obtained is well-known and can be found for instance in \cite{tsfasman1991algebraic}. In the case of rank-metric codes, the result in Theorem \ref{th:connection} is  stronger than the ones described in \cite{sheekey2019scatterd,Randrianarisoa2020ageometric,alfarano2021linear}. Indeed, in these papers the authors only shown \eqref{eq:weight_dimension}, while here we provide a deeper duality connection between the supports of codewords and the associated hyperplanes sections. 
 
 Actually, evidence of the existence of a duality can  already be found in \cite[Theorem 5.6]{alfarano2021linear}. However, the duality was not explicitly derived, and this result immediately follows now from Theorem \ref{th:connection}. 
\end{remark}

\begin{definition}
An \textbf{$\Fmnkd$ system} $\mathcal{U}$ is an ordered set $(\mathcal{U}_1,\ldots,\mathcal{U}_t)$, where, for any $i\in [t]$, $\mathcal{U}_i$ is an $\F_q$-subspace of $\F_{q^m}^k$ of dimension $n_i$, such that
$ \langle \mathcal{U}_1, \ldots, \mathcal{U}_t \rangle_{\F_{q^m}}=\F_{q^m}^k$ and 
$$ d=N-\max\left\{\sum_{i=1}^t\dim_{\F_q}(\mathcal{U}_i\cap \mathcal H) \st \mathcal H \textnormal{ is an $\F_{q^m}$-hyperplane of }\F_{q^m}^k\right\}.$$
\end{definition}

\begin{definition}
 Let $\mU=(\mathcal{U}_1,\ldots,\mathcal{U}_t)$ be  an $\Fmnkd$ system. The \textbf{subspace section} of an $\Fm$-subspace $H\subseteq \Fm^k$ in $\mU$ is the $t$-uple
 $$\chi_{\mU}(H)\coloneqq (\mathcal{U}_1\cap H,\ldots, \mathcal{U}_t\cap H) $$
 and the \textbf{dimension-list} of $H$ in $\mU$ is the associated $t$-uple of the $\Fq$-dimensions, that is
 $\nu_{\mU}(H)\coloneqq (\dim_{\fq}(\mathcal{U}_1\cap H),\ldots,\dim_{\fq}(\mathcal{U}_t\cap H))$.
 The \textbf{dimension-profile} of $H$ in $\mU$ is the $t$-uple obtained from $\nu_{\mU}(H)$ by rearranging the entries in non-increasing order.
\end{definition}

\begin{definition}\label{def:equiv_systems}
Two $\Fmnkd$ systems $(\mathcal{U}_1,\ldots,\mathcal{U}_t)$ and $(\mathcal{V}_1,\ldots, \mathcal{V}_t)$ are \textbf{equivalent} if there exists an isomorphism $\varphi\in\GL(k,\F_{q^m})$, an element $\mathbf{a}=(a_1,\ldots,a_t)\in (\Fm^*)^t$ and a permutation $\sigma\in\mathcal{S}_t$, such that for every $i\in[t]$
$$ \varphi(\mathcal{U}_i) = a_i\mathcal{V}_{\sigma(i)}.$$
We denote the set of equivalence classes of $\Fmnkd$ systems by $\mathfrak{U}\Fmnkd$.
\end{definition}

\begin{remark}
Observe that when $t=1$, the notion of $\Fmnkd$ system coincides with the notion of $[n,k,d]_{q^m/q}$ system (or $q$-system) introduced in \cite{Randrianarisoa2020ageometric}; see also \cite{alfarano2021linear}. The case when $n_1=\ldots=n_t=1$ is a bit more complicated, but we see now that it coincides  with the notion of projective systems only when considering the equivalence classes. Indeed, a $[\mathbf{1},k,d]_{q^m/q}$ system is not a projective $[t,k,d]_{q^m}$ system as defined for instance in \cite{tsfasman1991algebraic}. The difference is that projective systems are sets of points in $\PG(k-1,q^m)$ (or equivalently $1$-dimensional $\Fm$-subspaces of $\Fm^k$), while  $[\mathbf{1},k,d]_{q^m/q}$ systems are \textit{ordered} sets of $1$-dimensional $\Fq$-subspaces of $\Fm^k$. However, when considering the equivalence classes, they coincide, because we identify $\Fm$-multiples and we take into account the action of the symmetric group $\mathcal S_t$.
In other words, $\mathfrak U[\mathbf{1},k,d]_{q^m/q}$ can be identified with the set of equivalence classes of projective $[t,k,d]_{q^m}$ systems.
\end{remark}

Now  we  describe  the  1-to-1  correspondence  between  equivalence  classes  of  sum-rank nondegenerate $\Fmnkd$ code and equivalence classes of $\Fmnkd$ systems.

Consider the following maps
\begin{align*}
    \Psi :  \mathfrak{C}\Fmnkd &\to\mathfrak{U}\Fmnkd \\
    \Phi : \mathfrak{U}\Fmnkd &\to \mathfrak{C}\Fmnkd.
\end{align*}
Let $[\C]\in\mathfrak{C}\Fmnkd$ where $\C$ is a nondegenerate $\Fmnk$ code. Let $G=(G_1 \,\mid\, \ldots \,\mid\, G_t)$ be a generator matrix for $\cC$. Define $\Psi([\C])$ as the equivalence class of $\Fmnkd$ systems $[\mathcal{U}]$, where $\mathcal{U}=(\mathcal{U}_1,\ldots,\mathcal{U}_t)$ and $\mathcal{U}_i$ is the $\F_q$-span of the columns of $G_i$ for every $i\in[t]$. Viceversa, given $[(\mathcal{U}_1,\ldots,\mathcal{U}_t)]\in\mathfrak{U}\Fmnkd$, for every $i\in[t]$, fix an $\F_q$-basis $\{g_{i,1}, \ldots, g_{i,n_i}\}$ of $\mathcal{U}_i$. Define $G_i$ as the matrix whose columns are $\{g_{i,1}, \ldots, g_{i,n_i}\}$ and let $\Phi([(\mathcal{U}_1,\ldots,\mathcal{U}_t)])$ be the equivalence class of the sum-rank metric code generated by $G=(G_1 \,\mid\, \ldots \,\mid\, G_t).$

We can actually see that these two maps are well-defined. This is due to the two definitions of equivalence for $\Fmnkd$ codes (Definition \ref{def:equiv_codes}) and $\Fmnkd$ systems (Definition \ref{def:equiv_systems}). First, if we choose two different generator matrices for the same $\Fmnkd$ code $\C$, then it is clear that the two $\Fmnkd$ systems spanned by the columns of these matrices are equivalent, since one matrix is obtained from the other by left multiplication of an element in $\GL(k,q^m)$. Moreover, it is also immediate to see that if two $\Fmnkd$ codes $\C, \C'$ are equivalent, then the systems associated to any generator matrices $G$ and $G'$ are equivalent, showing that the map $\Psi$ is well-defined. A similar argument can be used to prove that also $\Phi$ is well-defined. 

Thus we can finally state the corresponding result between $\Fmnkd$ codes and $\Fmnkd$ systems. The proof is based on the above discussions and the details are left to the reader.

\begin{theorem}\label{thm:1-1corr}
$\Psi$ and $\Phi$ are well-defined and are the inverse of each other.
\end{theorem}

As a consequence, we can restate Theorem \ref{th:connection} in a more compact form, highlighting the duality identified by the sum-rank support of a codeword and the subspace section of an hyperplane as follows.

\begin{corollary}[Duality supports-sections]\label{cor:duality_supports_sections}
Let $\C$ be a nondegenerate $\Fmnk$ code and let $\mU \in \Psi([\C])$. Then, for every $G=(G_1 \,\mid\, \ldots \,\mid\, G_t)$ generator matrix of $\C$ and every $v\in \Fm^k$, we have 
$$\chi_{\mU}(v^\perp)=\psi_G(\supp_{\bfn}(vG)^\perp).$$
\end{corollary}

\begin{remark}
If $n_1=\ldots=n_t=1$, up to the identification of $\mathfrak U[\mathbf{1},k,d]_{q^m/q}$ with the set of equivalence classes of projective $[t,k,d]_{q^m}$ systems, $\Phi$ and $\Psi$ define exactly the relation between equivalence classes of projective $[t,k,d]_{q^m}$ systems and equivalence classes of nondegenerate Hamming-metric codes, see \cite[Theorem 1.1.6]{tsfasman1991algebraic}. In this case, we denote by $\Phi^{\HH}$ and $\Psi^{\HH}$ the maps $\Phi$ and $\Psi$, respectively.
\end{remark}

Finally, from Theorem \ref{th:connection} and Theorem \ref{th:Singletonbound}, we can rephrase in geometric terms the MSRD property of a code.

\begin{corollary} \label{teo:designMSRD} 
Let $\C$ be a sum-rank nondegenerate $\Fmnkd$ code and let $[(\mathcal{U}_1,\ldots,\mathcal{U}_t)]=\Psi([\mathcal{C}])$. $\mathcal{C}$ is an MSRD if and only if 
\[
\max\left\{ \sum_{i=1}^t \dim_{\fq}(\mathcal{U}_i \cap \mathcal{H})  \st \mathcal{H} \mbox{ hyperplane of }\F_{q^m}^k  \right\} \leq k-1.
\]
\end{corollary}

The notion of $\Fmnkd$ system is therefore dual to the notion of $\Fmnkd$ code. Despite these objects are equivalent, the geometric point of view as $\Fmnkd$ systems naturally captures the block structure of the metric space. This can be very convenient, depending on the properties that one aims to study of sum-rank metric codes. We will indeed use this geometric characterization later in the next section with the principal goal of constructing one-weight codes, a problem that seems to us more difficult to deal with in the coding theory setting. 
This is not surprising, since there is plenty of literature in which the geometric viewpoint as projective systems allows more naturally to derive results on Hamming-metric codes. 

\section{Doubly-extended linearized Reed-Solomon codes}\label{sec:4}

A particular interest among sum-rank metric is the family of MSRD codes.  The first construction of  MSRD codes was due to Martínez-Peñas in \cite{martinezpenas2018skew} which is now known as \emph{linearized Reed-Solomon codes}. These codes extend to the sum-rank metric the family of Gabidulin codes in the rank metric and the Reed-Solomon codes in the Hamming metric. They are introduced originally in a vectorial setting while in \cite{neri2021twisted} they have been revisited in a skew polynomial setting.

To present the family of linearized Reed-Solomon codes we will need the to introduce 
skew polynomials, which were originally presented in the seminal paper by Ore \cite{ore1933theory}, in a
quite general setting.

Let $\sigma$ be a generator of $\Gal(\F_{q^m}/\F_q)$.
The ring of skew polynomials $\F_{q^m}[x;\sigma]$ is a set of formal polynomials
\[
\left\{f=\sum_{i=0}^d f_i x^i \st f_i \in \F_{q^m}, d \in \Z_{\geq 0} \right\}
\]
equipped with the ordinary (component-wise) addition 
\[
f+g=\sum_{i \geq 0}(f_i+g_i)x^i
\]
and the multiplication rule 
\[
x \cdot a=\sigma(a) \cdot x, \mbox{ for any } a \in \F_{q^m},
\]
extended to polynomials by associativity and distributivity.

The degree of a skew polynomial is defined as in the commutative case, and for an integer $v$, we denote by $\F_{q^m}[x; \sigma]_{<v}$ the set of skew polynomials of degree smaller than $v$.

An element $\beta \in \F_{q^m}$ is a root of a skew polynomial $f=\sum_{i=0}^d f_i x^i$, if $f(\beta):=\sum_{i=0}^d f_i \sigma^i(\beta)=0$. The set of roots of $f$ over $\F_{q^m}$ is denoted by $\ker(f)$ and it turns out to be an $\fq$-subspace of $\F_{q^m}$ having dimension at most $d$, see \cite{lidl1997finite} and \cite[Lemma 3.2]{guralnick1994invertible}.
We are interested in studying sum-rank metric codes in the vectorial setting (as in Subsection \ref{sub:vectorsetting}). To this aim, we recall the \emph{generalized operator evaluation} defined in \cite{leroy1995pseudolinear} as follows. Let $f \in \F_{q^m}[x; \sigma]$ and $a \in \F_{q^m}$. Define the map
\[
\begin{array}{rccl}
f(\cdot)_a: & \F_{q^m}  & \longrightarrow & \F_{q^m} \\
& \beta & \longmapsto & \sum_{i \geq 0} f_i \sigma^i(\beta)\N_i(a),
\end{array} 
\]
where $\N_i(a)\coloneqq\prod_{j=0}^{i-1}\sigma^j(a)$. 

Let $a=(a_1,\ldots,a_t) \in \F_{q^m}^t$ and $\beta=(\beta_1,\ldots,\beta_n) \in \F_{q^m}^n$. We define the multi-point evaluation map as 

\[
\begin{array}{rccl}
\mathrm{ev}_{a,\beta}: & \F_{q^m}[x;\sigma] & \longrightarrow & \F_{q^m}^\bfn \\
&f  & \longmapsto &  (f(\beta_1)_{a_1},\ldots,f(\beta_n)_{a_1} \lvert   \cdots \lvert  f(\beta_1)_{a_t},\ldots, f(\beta_n)_{a_t}).
\end{array} 
\]

In particular, we will use the multi-point evaluation map under certain assumptions on $a$ and $\beta$.
 
\begin{definition}
A pair $(a,\beta) \in \F_{q^m}^t \times \F_{q^m}^n$ is said to be an \textbf{evaluation pair} (with respect to $\sigma$) if the components of $a$ have pairwise distinct norm over $\fq$, and the elements of $\beta$ are linearly independent over $\F_q$. \end{definition}

The properties of the map $\mathrm{ev}_{a,\beta}(\cdot)$ have been exploited in \cite{martinezpenas2018skew} and they lead to the construction of the the following class of MSRD codes.

\begin{definition}[see \textnormal{\cite[Definition 31]{martinezpenas2018skew}}]\label{def:linRScodes}
Let $n \leq m$, $t \leq q-1$ and $1 \leq k \leq N$. Let $(a,\beta) \in \F_{q^m}^t \times \F_{q^m}^n$ be an evaluation pair. The corresponding \textbf{linearized Reed-Solomon code} is defined by
\[
\mathcal{C}_{k,a,\beta}=\{ \mathrm{ev}_{a,\beta}(f) \st f \in \F_{q^m}[x;\sigma]_{<k} \}.
\]
\end{definition}

In \cite[Theorem 4]{martinezpenas2018skew}, it has been shown that the linearized Reed-Solomon codes $\mathcal{C}_{k,a,\beta}$ are MSRD codes.
For a description of the duals of linearized Reed-Solomon codes, see
\cite{caruso2021duals,caruso2021atheory}.

Hence, using Corollary \ref{teo:designMSRD}, we can look at the intersections of an $\Fmnkd$ system associated with a linearized Reed-Solomon code with the hyperplanes.

\begin{proposition}\label{prop:LRScodes}
Let $k \leq N$, $t<q$ and $\mathcal{S}=\langle \beta_1,\ldots,\beta_n \rangle_{\fq}$. Let $\mathcal{C}_{k,a,\beta}$ be a linearized Reed-Solomon code as in Definition \ref{def:linRScodes}. Let $[(\mathcal{U}_1,\ldots,\mathcal{U}_t)]=\Psi([\mathcal{C}_{k,a,\beta}])$. Then a  representative of the class $[(\mathcal{U}_1,\ldots,\mathcal{U}_t)]$ is
\[ \mathcal{U}_i=\{ (y,\sigma(y)\N_1(a_i),\sigma^2(y)N_2(a_i),\ldots, \sigma^{k-1}(y)N_{k-1}(a_i)) \st y \in \mathcal{S}\}. \]
In particular,
    \[\sum_{i=1}^t \dim_{\fq}(\mathcal{U}_i \cap \mathcal{H})   \leq k-1, \]
    for every hyperplane $\mathcal{H}$ of $\F_{q^m}^k$.
\end{proposition}
\begin{proof}
A generator matrix $G$ of $\mathcal{C}_{k,a,\beta}$ can be computed as in \cite[p. 604]{martinezpenas2018skew}. Then it is easy to see that the $\F_q$-span of the columns of each block of $G$ gives the $\F_q$-subspace $\mathcal{U}_i$.
The last part then follows by Corollary \ref{teo:designMSRD}.
\end{proof}

As it happens for Reed-Solomon codes in the Hamming metric, we now want to extend the evaluations of a linearized Reed-Solomon code also to $0$ and to $\infty$. For this reason,
for a skew polynomial $f(x)=f_0+f_1x+\ldots+f_{k-1}x^{k-1} \in \F_{q^m}[x; \sigma]_{<k}$  define $f(\beta)_{\infty}\coloneqq\beta f_{k-1}$ for any $\beta \in \mathbb{F}_{q^m}$.

\begin{definition}
Let $1 \leq k \leq N$, $\gamma,\delta \in \F_{q^m}^*$  and let $(a,\beta) \in \F_{q^m}^t \times \F_{q^m}^n$ be an evaluation pair. We define the code
\[
\mathcal{C}_{k,a,\beta}^{\gamma,\delta}=\{ (\mathrm{ev}_{a,\beta}(f)\lvert f(\gamma)_{0} \lvert f(\delta)_{\infty}) \st f \in \F_{q^m}[x;\sigma]_{<k} \}.
\]
When $t=q-1$, the code $\mathcal{C}_{k,a,\beta}^{\gamma,\delta}$ is called the \textbf{doubly-extended linearized Reed-Solomon} code.
\end{definition}

\begin{remark}
Note that if $m=1$, then we recover the notion of doubly-extended Reed-Solomon codes. 
\end{remark}

As one might expect, also doubly-extended linearized Reed-Solomon codes are MSRD.

\begin{theorem}\label{th:MSRDLRS}
Let  $t \leq q-1$ and $2 \leq k \leq N$. Let $(a,\beta) \in \F_{q^m}^t \times \F_{q^m}^n$ be an evaluation pair, let $\gamma,\delta \in \F_{q^m}^*$ and let $\mathcal{S}=\langle \beta_1,\ldots,\beta_n \rangle_{\fq}$. Then 
\[ [(\mathcal{U}_1,\ldots,\mathcal{U}_t,\mathcal{U}_0,\mathcal{U}_{\infty})]= 
\Psi([\mathcal{C}_{k,a,\beta}^{\gamma,\delta}]) . 
\]
where \[ \mathcal{U}_i=\{ (y,\sigma(y)\N_1(a_i),\sigma^2(y)N_2(a_i),\ldots, \sigma^{k-1}(y)N_{k-1}(a_i)) \st y \in \mathcal{S}\},\]
for any $i \in [t]$, $\mathcal{U}_0=\langle (\gamma,0,\ldots,0) \rangle_{\fq}$ and $\mathcal{U}_{\infty}=\langle (0,\ldots,0,\delta) \rangle_{\fq}$.
Moreover, the code $\mathcal{C}_{k,a,\beta}^{\gamma,\delta}$ is an MSRD code.
\end{theorem}
\begin{proof}
By Proposition \ref{prop:LRScodes} we have $[(\mathcal{U}_1,\ldots,\mathcal{U}_t)]=\Psi([\mathcal{C}_{k,a,\beta}])$, and hence $[(\mathcal{U}_1,\ldots,\mathcal{U}_t,\mathcal{U}_0,\mathcal{U}_{\infty})]=\Psi([\mathcal{C}_{k,a,\beta}^{\gamma,\delta}])$.
Now, consider any $\mathbb{F}_{q^m}$-hyperplane $\mathcal{H}$ of $\mathbb{F}_{q^m}^k$ and let $h_1x_1+\cdots+h_kx_k=0$ be an its equation.
Consider
\[
g_{i,\mathcal{H}}=h_1 +h_2\N_1(a_i)x+h_3N_2(a_i)x^2+\ldots+ h_kN_{k-1}(a_i)x^{k-1} \in \F_{q^m}[x;\sigma].
\]
Since $\dim_{\fq}(\mathcal{U}_i \cap \mathcal{H}) = \dim_{\fq} (\ker(g_{i,H}) \cap \mathcal{S})$, by Corollary \ref{teo:designMSRD} one gets
\begin{equation}\label{eq:condLRS}
\sum_{i=1}^t \dim_{\fq}(\mathcal{U}_i \cap \mathcal{H}) = \sum_{i=1}^t \dim_{\fq} (\ker(g_{i,\mathcal{H}}) \cap \mathcal{S}) \leq k-1.
\end{equation}

Moreover, by Corollary \ref{teo:designMSRD} it is enough to prove that 
\[
\sum_{i=1}^t \dim_{\fq}(\mathcal{U}_i \cap \mathcal{H}) +\dim_{\fq}(\mathcal{U}_0 \cap \mathcal{H})+\dim_{\fq}(\mathcal{U}_{\infty} \cap \mathcal{H}) \leq k-1,
\]
for every $\mathbb{F}_{q^m}$-hyperplane $\mathcal{H}$ of $\mathbb{F}_{q^m}^k$.

Clearly, if $\dim_{\fq}(\mathcal{U}_0 \cap \mathcal{H})+\dim_{\fq}(\mathcal{U}_{\infty} \cap \mathcal{H})=0$, then the assertion follow by \eqref{eq:condLRS}.

Suppose now that $\dim_{\fq}(\mathcal{U}_0 \cap \mathcal{H})=1$ and $\dim_{\fq}(\mathcal{U}_{\infty} \cap \mathcal{H})=0$.
The equation of $\mathcal{H}$ is $h_2x_2+\cdots+h_kx_k=0$. 
By \eqref{eq:condLRS},
\[
\sum_{i=1}^t \dim_{\fq}(\mathcal{U}_i \cap \mathcal{H})=\sum_{i=1}^t \dim_{\fq}(\ker(g'_i)\cap \mathcal{S})\leq k-2,
\]
with $g'_i=h'_1 \mathrm{id}+h'_2 \N_1(a_i)\sigma+h_3N_2(a_i)\sigma^2+\cdots+ h'_{k-1}N_{k-2}(a_i)\sigma^{k-2}$, where $h'_i=\sigma^{m-1}(h_{i+1})$, for any $i \in [k-1]$.

Similarly, the case $\dim_{\fq}(\mathcal{U}_0 \cap \mathcal{H})=0$ and $\dim_{\fq}(\mathcal{U}_{\infty} \cap \mathcal{H})=1$ can be carried out to get the assertion.

Finally, suppose that $\dim_{\fq}(\mathcal{H}\cap \mathcal{U}_0)+\dim_{\fq}(\mathcal{H}\cap \mathcal{U}_{\infty})=2$, then
the equation of $\mathcal{H}$ is $h_2x_2+\ldots+h_{k-1}x_{k-1}=0$. Then, by \eqref{eq:condLRS}, it follows
\[
\sum_{i=1}^t \dim_{\fq}(\mathcal{U}_i \cap \mathcal{H})=\sum_{i=1}^t \dim_{\fq}(\ker(g''_i)\cap \mathcal{S})\leq k-3,
\]
with $g''_i=h''_1 \mathrm{id}+h''_2\N_1(a_i)\sigma+h_3''N_2(a_i)\sigma^2+\cdots+ h''_{k-2}N_{k-3}(a_i)\sigma^{k-3}$, where $h''_i=\sigma^{m-1}(h_{i+1})$ for any $i\in[k-2]$.
\end{proof}

In the following lemma we study some properties of the $[\bfn,2,d]_{q^m/q}$ system associated with the code $\mathcal{C}_{2,a,\beta}^{\gamma,\delta}.$

\begin{lemma}\label{lem:disjointpseudo}
Let $m$ be a positive integer.
Let $a_1,\ldots,a_{q-1} \in \F_{q^m}^*$ such that $\N_{q^m/q}(a_i) \neq \N_{q^m/q}(a_j)$, if $i \neq j$. Let $\mathcal{U}_i=\{(y,a_i \sigma(y))\st y \in \F_{q^m}\}$, for $i\in [q-1]$. Let $\gamma,\delta \in \F_{q^m}^*$ and define $\mathcal{U}_0=\{(a\gamma,0)\st a \in \F_q\}$ and $\mathcal{U}_{\infty}=\{(0,b\delta) \st b \in \F_q\}$. Then for every one-dimensional $\F_{q^m}$-subspace $\mathcal{S}$ of $\mathbb{F}_{q^m}^2$ there exists $i \in \{0,\ldots,q-1,\infty\}$ such that $\dim_{\fq}(\mathcal{S}\cap \mathcal{U}_i)=1$ and $\dim_{\fq}(\mathcal{S}\cap \mathcal{U}_j)=0$ for every $j \ne i$.
\end{lemma}
\begin{proof}
Let $\mathcal{S}$ be any one-dimensional $\F_{q^m}$-subspace of $\mathbb{F}_{q^m}^2$.
If $\mathcal{S}\in \{\langle(1,0)\rangle_{\F_{q^m}},\langle(0,1)\rangle_{\F_{q^m}}\}$, then either $\dim_{\fq}(\mathcal{S}\cap \mathcal{U}_0)=1$ or $\dim_{\fq}(\mathcal{S} \cap \mathcal{U}_{\infty})=1$. In both the cases, the intersection with remaining $\mathcal{U}_i$'s will be trivial.
Suppose now that $\mathcal{S}\notin \{\langle(1,0)\rangle_{\F_{q^m}},\langle(0,1)\rangle_{\F_{q^m}}\}$. Then $\mathcal{S}=\langle(1,\lambda)\rangle_{\F_{q^m}}$ with $\lambda \in \F_{q^m}^*$ and, by the assumptions on the $a_i$'s it follows that there exists $i \in [q-1]$ such that $\N_{q^m/q}(\lambda)=\N_{q^m/q}(a_i)$.
Hence, there exists an element $y \in \F_{q^m}^*$ such that $\lambda=a_i \sigma(y)y^{-1}$. This implies that
\[ \langle (1,\lambda)\rangle_{\F_{q^m}}=\langle (y,a_i \sigma(y))\rangle_{\F_{q^m}}, \]
and hence $\dim(\mathcal{U}_i\cap \langle (1,\lambda)\rangle_{\F_{q^m}})>0$.
Moreover, if $y \in \F_{q^m}^*$ and $\rho \in \F_{q^m}^*$ are such that $\rho(1,\lambda)=(y,a_i\sigma(y))$ then
\[ \left\{
\begin{array}{ll}
\rho=y,\\
\rho \lambda = a_i \sigma(y),
\end{array}
\right. \]
from which we get $\dim_{\fq}(\mathcal{U}_i\cap \langle (1,\lambda)\rangle_{\F_{q^m}})=1$.
Let $j \in [q-1]$ with $j\ne i$. 
Since $\N_{q^m/q}(\lambda)\ne \N_{q^m/q}(a_j)$, if $\dim_{\fq}(\mathcal{U}_j \cap \langle (1,\lambda)\rangle_{\F_{q^m}})>0$ then, arguing as before, there exists $y \in \F_{q^m}^*$ and $\rho \in \F_{q^m}^*$ such that $\rho(1,y)=(y,a_j\sigma(y))$ and so $y\lambda = a_j\sigma(y)$, which implies that $\N_{q^m/q}(\lambda) = \N_{q^m/q}(a_j)$, a contradiction. The assertion then follows since $\dim_{\fq}(\mathcal{S}\cap \mathcal{U}_0)=\dim_{\fq}(\mathcal{S}\cap \mathcal{U}_{\infty})=0$.
\end{proof}

\begin{remark}
From a geometric point of view, we can consider the linear sets associated with the $\mathcal{U}_i$'s  of Lemma \ref{lem:disjointpseudo} (see Section \ref{sec:extendHamm} and in particular Definition \ref{def:linearset} for the notion of linear sets).
They turn out to be pairwise disjoint scattered linear sets covering all the projective line, that is $L_{\mathcal{U}_0} \cup \cdots \cup L_{\mathcal{U}_{\infty}}=\PG(1,q^m)$, see the next section.
The linear sets $L_{\mathcal{U}_i}$ with $i \in \{1,\ldots,q-1\}$ are called of \textbf{pseudoregulus type}, see \cite{lunardon2014maximum}.
\end{remark}

Lemma \ref{lem:disjointpseudo} essentially implies that the $2$-dimensional doubly-extended linearized Reed-Solomon codes are  one-weight codes, exactly  as it happens for doubly-extended Reed-Solomon codes in the Hamming metric. 

\begin{theorem}\label{th:RS1w}
Let $n = m$, $t = q-1$, $k=2$ and let $(a,\beta) \in \F_{q^m}^t \times \F_{q^m}^m$ be an evaluation pair.
Let $\gamma,\delta \in \F_{q^m}^*$, then the code
$\mathcal{C}_{2,a,\beta}^{\gamma,\delta}$ is a one-weight MSRD code.
\end{theorem}
\begin{proof}
Consider
$\mathcal{U}_i=\{(y,a_i \sigma(y))\st y \in \F_{q^m}\}$, for any $i\in[q-1]$ and the $a_i$'s have pairwise distinct norm over $\fq$, $\mathcal{U}_0=\langle(\gamma,0)\rangle_{\fq} $ and $\mathcal{U}_{\infty}=\langle(0,\delta)\rangle_{\fq}$.
Then by Theorem \ref{th:MSRDLRS} $[(\mathcal{U}_1,\ldots,\mathcal{U}_t,\mathcal{U}_0,\mathcal{U}_{\infty})]=\Psi([\mathcal{C}_{2,a,\beta}^{\gamma,\delta}])$. 
By Lemma \ref{lem:disjointpseudo}, since the $\mathcal{U}_i$'s meets every one-dimensional $\F_{q^m}$-subspace of $\F_{q^m}^2$ only in a one-dimensional $\fq$-subspace, then $\sum_{i=0}^{q-1} \dim_{\fq}(\mathcal{U}_i \cap S)+\dim_{\fq}(\mathcal{U}_{\infty}\cap S)=1$, for each one-dimensional $\F_{q^m}$-subspace $S$ in $\F_{q^m}^2$.
The assertion then follows by Theorem \ref{th:connection}.
\end{proof}

\section{Associating an Hamming-metric code} \label{sec:extendHamm}

The characterization of one-weight codes in a given metric has always been a fascinating problem, due to their interesting combinatorial properties and their strong rigidity structures. One-weight codes in the Hamming metric have been classified in the well-known work of Bonisoli \cite{bonisoli1983every}. An easy way to prove such a result is using the geometric characterization of Hamming-metric codes via projective systems. The geometric approach was the key used very recently to obtain a classification of one-weight (vector) codes in the rank metric in \cite{Randrianarisoa2020ageometric}; see also \cite[Proposition 3.16]{alfarano2021linear}. From a more intrinsic point of view, this approach can be also seen as relying on Bonisoli's result for Hamming-metric codes by using the transformation introduced in \cite[Section 4]{alfarano2021linear} to obtain a Hamming-metric code from a rank-metric one. Unfortunately, this idea seems unlikely to work if we try to obtain a similar classification result for  one-weight sum-rank metric codes. However, we can use a similar strategy to obtain partial results on sum-rank metric codes with constant rank-profile.

To this aim, we will recall some basics on linear sets.

\begin{definition}\label{def:linearset}
Let $\Lambda=\PG(\F_{q^m}^k,\F_{q^m})=\PG(k-1,q^m)$.
Let $\mathcal{U}$ be an $[n,k]_{q^m/q}$ system, then the set of points
\[ L_\mathcal{U}=\{\la {u} \ra_{\mathbb{F}_{q^m}} \st  { u}\in \mathcal{U}\setminus \{{\bf 0} \}\}\subseteq \Lambda \]
is said to be an $\fq$-\textbf{linear set of rank} $n$.
\end{definition}

For such a linear set we have the following bound on its size
\begin{equation}\label{eq:boundsizels}
    |L_\mathcal{U}| \leq \frac{q^n-1}{q-1}.
\end{equation}

Let $P=\langle v\rangle_{\F_{q^m}}$ be a point in $\Lambda$. The \emph{weight of $P$ in $L_\mathcal{U}$} is defined as 
\[ \ww_{L_\mathcal{U}}(P)=\dim_{\fq}(\mathcal{U}\cap \langle v\rangle_{\F_{q^m}}). \]
If $L_\mathcal{U}$ is an $\fq$-linear set with the property that $\ww_{L_\mathcal{U}}(P)\leq 1$ for every point $P \in \Lambda$, we say that $L_\mathcal{U}$ is \emph{scattered}.
Equivalently, $L_\mathcal{U}$ is scattered if and only if the size of $L_\mathcal{U}$ satisfies the equality in \eqref{eq:boundsizels}, that is $|L_\mathcal{U}|=\frac{q^n-1}{q-1}$. 
Moreover, in \cite{blokhuis2000scattered} it has been proved that for scattered $\fq$-linear sets of rank $n$ in $\Lambda$ we have
\begin{equation}\label{eq:boundscatt}
    n\leq \frac{mk}2.
\end{equation}
For more details on linear sets see \cite{lavrauw2015field,polverino2010linear}.

\subsection{The associated Hamming-metric code}
Following the strategy developed in \cite[Section 4]{alfarano2021linear}, we associate an equivalence class of Hamming-metric codes to an equivalence class of sum-rank-metric codes. For consistency, we will keep as much as possible the notation introduced in \cite{alfarano2021linear}.

Denote by 
$\mathfrak D\Fmkt$ the set of $n$-dimensional $\Fq$-subspaces of $\Fm^k$, and let $\mathfrak D(\frac{q^n-1}{q-1},k)_{q^m}$ be the set  of multisets of $\PG(k-1,q^m)$ with cardinality (counted with multiplicity) $\frac{q^n-1}{q-1}$. Define the map
$$\begin{array}{rccc}\Ext^{1}:& \mathfrak D\Fmkt & \longrightarrow & \mathfrak D(\frac{q^n-1}{q-1},k)_{q^m} \\
&\mU & \longmapsto & (L_\mU, \mm_\mU),\end{array}$$
where $\mm_\mU$ is the multiplicity function given by
$$\mm_\mU(\langle v\rangle_{\Fm})\coloneqq\frac{q^{\dim_{\Fq}(\mU\cap \langle v\rangle_{\Fm})}-1}{q-1}.$$
For $i \in [t]$, let $(\cM_i,\mm_{i})\in \mathfrak D(n_i,k)_{q^m}$ be a collection of multisets of $\PG(k-1,q^m)$. Define their \textbf{disjoint union} as
$$\biguplus_{i=1}^t (\cM_i,\mm_{i})\coloneqq(\cM,\mm),$$
where $\cM=\cM_1\cup\ldots\cup \cM_t$, and $\mm(P)=\mm_1(P)+\ldots+\mm_t(P)$ for every $P\in \PG(k-1,q^m)$.

Let $\mathfrak D(\bfn,k)_{q^m/q}$ be the set of $\Fmnk$ systems. Define the map
 $$\begin{array}{rccc}\Ext:&\mathfrak D(\bfn,k)_{q^m/q} & \longrightarrow & \mathfrak D(\frac{q^{n_1}+\ldots +q^{n_t}-t}{q-1},k)_{q^m} \\
& \mU\coloneqq(\mU_1,\ldots,\mU_t) & \longmapsto & \biguplus\limits_{i=1}^t\Ext^1(\mU_i).\end{array}$$

\begin{proposition}\label{prop:ext_welldefined}\;
\begin{enumerate}
    \item If $\mU$ is an $\Fmnk$ system, then $\Ext(\mU)$ is a  projective $[\frac{q^{n_1}+\ldots +q^{n_t}-t}{q-1},k]_{q^m}$ system.
    \item If $\mU, \cV$ are two equivalent $\Fmnk$ systems, then $\Ext(\mU)$ and $\Ext(\cV)$ are two equivalent projective $[\frac{q^{n_1}+\ldots +q^{n_t}-t}{q-1},k]_{q^m}$ systems.
\end{enumerate}
\end{proposition}

\begin{proof}
\begin{enumerate}
    \item Let $\mU\coloneqq(\mU_1,\ldots,\mU_t)$ be an $\Fmnk$ system. Then, by definition it holds that $\langle \mathcal{U}_1,\ldots,\mathcal{U}_t\rangle_{\Fm}=\Fm^k$. Thus, in $\PG(k-1,q^m)$ we have 
    $$\langle \Ext(\mU)\rangle=\langle L_{\mU_1},\ldots,L_{\mU_t}\rangle=L_{\langle \mU_1,\ldots,{\mU_t}\rangle_{\Fm}}=\PG(k-1,q^m). $$
    \item Let $\mU\coloneqq(\mU_1,\ldots,\mU_t)$, $\mV\coloneqq(\mV_1,\ldots,\mV_t)$ be two equivalent $\Fmnk$ systems. Up to permutation, we may assume that there exists $f \in \GL(k,q^m)$ such that $f(\mU_i)=\mV_i$, for every $i \in [t]$. Therefore, we have that 
    \begin{align*}\Ext(\mV)&=\biguplus\limits_{i=1}^t\Ext^1(\mV_i) = \biguplus\limits_{i=1}^t\Ext^1(f(\mU_i))= \biguplus\limits_{i=1}^t(L_{f(\mU_i)},\mm_{f(\mU_i)}) \\
    &=\biguplus\limits_{i=1}^t(\overline{f}(L_{\mU_i}),\mm_{f(\mU_i)}) = \overline{\phi}(\Ext(\mU)),
    \end{align*}
    where $\overline{f}$ is the projection of $f$ in $\PGL(k,q^m)$, proving the claim.
\end{enumerate}
\end{proof}

Hence, by Proposition \ref{prop:ext_welldefined}, we have a natural induced map 
$$ \Ext^{\HH}: \mathfrak U \Fmnk \longrightarrow \mathfrak U[M,k]_{q^m},$$
where $M\coloneqq \frac{q^{n_1}+\ldots+q^{n_t}-t}{q-1}$, and  $\mathfrak U[M,k]_{q^m}$ denotes the set of the equivalence classes of projective $[M,k]_{q^m}$ systems,
\begin{definition}
Let $\C$ be a sum-rank nondegenerate $\Fmnk$ code. Any code $C \in (\Phi^{\HH}\circ \Ext^{\HH} \circ \Psi)([\C])$ is called an \textbf{associated Hamming-metric code} with $\C$. 
\end{definition}

When the choice of the associated Hamming-metric code with an $\Fmnk$ code $\C$ is not important nor relevant (e.g. when the property we are investigating is invariant under code equivalence), we will simply write $\C^{\HH}$ to denote it.  

\begin{definition}
 Let $G=(G_1 \,|\, \ldots \,|\, G_t)\in \Fm^{k\times N}$ be the generator matrix of a nondegenerate $\Fmnk$ code, with $G_i \in \Fm^{k \times n_i}$, and let $\mU=(\mU_1,\ldots, \mU_t)$ be the $\Fmnk$ system where $\mU_i$  is the $\Fq$-span of the columns of $G_i$. We define 
 $G_{\Ext}\in \Fm^{k\times M}$ to be any matrix obtained from the projective $[M,k]_{q^m}$ system $\Ext(\mU)$, where $M\coloneqq \frac{q^{n_1}+\ldots+q^{n_t}-t}{q-1}$. 
\end{definition}

\begin{remark}
 In contrast with what happens for rank-metric codes (see \cite[Section 4]{alfarano2021linear}), in this case it is not clear how to determine the minimum distance of $\C^{\HH}$ starting from the minimum sum-rank distance of $\C$. However, the following result links the rank list distribution of $\C$ with the Hamming weights of $\C^{\HH}$.
\end{remark}

\begin{proposition}\label{prop:weight_G_Ext}
 Let $G=(G_1 \,|\, \ldots \,|\, G_t)\in \Fm^{k\times N}$ be the generator matrix of a nondegenerate $\Fmnk$ code, and let $v \in \Fm^k \setminus \{0\}$. Then
 $$\ww_{\HH}(vG_{\Ext})= \sum_{i=1}^t\frac{q^{n_i}-q^{n_i-\rk(vG_i)}}{q-1}=\sum_{i=1}^t\frac{q^{n_i}-q^{n_i-\rho(vG)_i}}{q-1}.$$

In particular, if $\C$ is a nondegenerate $\Fmnk$ code, then the minimum distance of $\C^{\HH}$ is given by
 $$ \dd(\C^{\HH})= \min_{\mathbf{r}\in \mathrm{S}(\C)}\sum_{i=1}^t\frac{q^{n_i}-q^{n_i-r_i}}{q-1},$$
 where $\mathrm{S}(\C)$ is the set of rank-lists of $\C$. 
\end{proposition}

\begin{proof}
Let $\mU=(\mU_1,\ldots, \mU_t)$ be the $\Fmnk$ system choosing as $\mU_i$ the $\Fq$-span of the columns of $G_i$. Let $\mathcal{H}$ be the hyperplane defined by $v^{\perp}$. Then Hamming-weight of $vG_{\Ext}$ is given by
\[
\mathrm{w}_{\HH}(vG_{\Ext})=M-\sum_{P \in \PG(\mathcal{H},\F_{q^m})} \mathrm{m}(P),
\]
where $\mathrm{m}$ is the multiplicity function of $\Ext(\mathcal{U})$. Observe that 
\begin{align*}
\sum_{P \in \PG(\mathcal H,\F_{q^m})} \mathrm{m}(P) & = \sum_{i=1}^t \sum_{P \in \PG(\mathcal H,\F_{q^m})} \mathrm{m_{\mathcal{U}_i}}(P) \\
& =\sum_{i=1}^t \sum_{\substack{\mathcal{V} \subseteq \mathcal{H}, \\ \dim_{\F_{q^m}}(\mathcal{V})=1}} \frac{q^{\dim_{\F_q}(\mathcal{U}_i \cap \mathcal{V})}-1}{q-1} \\
& =\frac{1}{q-1} \sum_{i=1}^t \lvert \mathcal{H} \cap (\mathcal{U}_i \setminus \{0\}) \rvert \\
& = \sum_{i=1}^t \frac{q^{\dim_{\F_q}(\mathcal{H} \cap \mathcal{U}_i )}-1}{q-1}.
    \end{align*}
Hence \[\mathrm{w}_{\HH}(vG_{\Ext})=\sum_{i=1}^t \frac{q^{n_i}-1}{q-1}-\sum_{i=1}^t \frac{q^{\dim_{\F_q}(\mathcal H \cap \mathcal{U}_i )}-1}{q-1}=
\sum_{i=1}^t\frac{q^{n_i}-q^{\dim_{\F_q}\mathcal (H \cap \mathcal{U}_i )}}{q-1}.
\]
The first part of the assertion follows from Theorem \ref{th:connection}, whereas the second one follows by definition of rank-list of $\mathcal{C}$.
\end{proof}

The above result extends to the sum-rank metric \cite[Theorem 8]{alfarano2021linear}, \cite[Section 5]{blokhuis2000scattered} and \cite[Theorem 7.1]{zini2020scattered}.

\subsection{Parameters of constant rank-profile codes}
We have seen how the machinery defined above links sum-rank metric codes to Hamming-metric codes. We now exploit this machinery to derive properties of constant rank-profile codes.

\begin{proposition}\label{prop:constant_rank_profile_simplex}
 Let $\C$ be a nondegenerate $[(n,\ldots,n),k]_{q^m/q}$ code which is constant rank-profile. Then $\C^{\HH}$ is a constant weight code in the Hamming metric. 
\end{proposition}

\begin{proof}
Let $\mathbf{r}=(r_1,\ldots,r_t)$ be the constant rank-profile. Let $G$ be a generator matrix of $\C$ and let us take the code $\C^{\HH}$ to be the one generated by $G_{\Ext}$. Then, for every $v \in \Fm^k\setminus \{0\}$, we have by Proposition \ref{prop:weight_G_Ext} that
$$ \ww_{\HH}(vG_{\Ext})=\sum_{i=1}^t\frac{q^{n}-q^{n-r_i}}{q-1}$$
is constant. 
\end{proof}

\begin{corollary}\label{cor:parameters}
 Let $\C$ be a nondegenerate $[(n,\ldots,n),k]_{q^m/q}$ code which is constant rank-profile, with rank-profile $\mathbf{r}=(r_1,\ldots, r_t)$. Then: 
 \begin{enumerate}
     \item The number $$\ell\coloneqq\frac{t(q^n-1)(q^m-1)}{(q-1)(q^{km}-1)}$$ is a positive integer.
     \item It holds that $$ t q^{m(k-1)} (q^n-1)(q^m-1)=(q^{km}-1)\Big(tq^n-\sum_{i=1}^tq^{n-r_i}\Big). $$
 \end{enumerate}
\end{corollary}

\begin{proof}
 By Proposition \ref{prop:constant_rank_profile_simplex}, we have that $\C^{\HH}$ is a constant weight code in the Hamming metric. However, the only nondegenerate constant weight codes in the Hamming metric are (equivalent to) concatenations of $\ell$ copies of the $[\frac{q^{km}-1}{q^m-1},k,q^{m(k-1)}]_{q^m}$ simplex code; see \cite{bonisoli1983every}. Since the length of $\C^{\HH}$ is $t\frac{(q^n-1)}{(q-1)}$, then we obtain that 
 $$ \ell=\frac{t(q^n-1)(q^m-1)}{(q-1)(q^{km}-1)}$$ is a positive integer.
 Moreover, by Proposition \ref{prop:weight_G_Ext}, the minimum distance of $\C^{\HH}$ must be  
 $$\sum_{i=1}^t\frac{q^n-q^{n-r_i}}{q-1}, $$
 leading to the desired equality.
\end{proof}

Corollary \ref{cor:parameters} provides constraints on the parameters that a nondegenerate constant profile $[(n,\ldots,n),k]_{q^m/q}$ code can have, as the following example will illustrate. 

\begin{example}\label{exa:constantrankprofile}
Let us assume that we want to construct a nondegenerate  $[(n,\ldots,n),k]_{q^m/q}$ code with constant rank-profile $\mathbf{r}=(r_1,\ldots,r_t)$, and let us fix $n=3$, $k=m=2$. Using Corollary \ref{cor:parameters}(1), we obtain that 
$$ t\frac{q^2+q+1}{q^2+1}$$
is an integer, and thus $t=(q^2+1)t'$.
Now, using Corollary \ref{cor:parameters}(2), we deduce after some simplifications that 
$$ t'(q^3+q^2)=\sum_{i=1}^{(q^2+1)t'}q^{3-r_i}.$$
On the other hand, observe that $3-r_i=\dim_{\Fq}(\mU_i\cap H)$ for some $\F_{q^2}$-hyperplane of $\F_{q^2}^2$. Since $\dim_{\Fq}(\mU_i)=3$, then we can only have $r_i \in \{1,2\}$. Let 
$a_i\coloneqq|\{ j \st  r_j=i \}|$, for $i=1,2$. Therefore, we have
$$\begin{cases} a_1+a_2=t'(q^2+1) \\
q^2a_1+qa_2=t'(q^3+q^2),
\end{cases}
$$
which leads to the unique solution
$$ a_1=t', \qquad a_2=t'q^2. $$
Thus, the constant rank-profile must be
$$ \mathbf{r}= ( \underbrace{2, \ldots, 2}_{t'q^2 \text{ {times}}}, \underbrace{1, \ldots, 1}_{t' \text{ {times}}}).$$
\end{example}

\section{Orbital construction and simplex codes}\label{sec:orbital}

In this section we provide a construction of constant rank-profile codes using the orbits of a transitive action on the nonzero elements of $\Fm^k$. When the orbits are obtained by the action of a Singer subgroup of $\mathrm{GL}(k,q^m)$, such codes will be called simplex codes.

\subsection{Orbital construction}

Let $\mathcal{G}\le\GL(k,q^m)$, and consider the action of $\mathcal{G}$ on $\Fm^k\setminus\{0\}$ -- which we denote by $\phi_{\mathcal{G}}$ -- induced by the one of $\GL(k,q^m)$, namely
$$ \begin{array}{rccl} \phi_{\mathcal{G}} : & \mathcal{G}\times (\Fm^k\setminus\{0\}) & \longrightarrow & \Fm^k\setminus\{0\}\\
&(A,v)& \longmapsto & vA. 
\end{array}$$
The same group $\mathcal{G}$ can be considered also acting on any Grassmannian. More precisely, for any positive integer $r$ dividing $m$, and any positive integer $s$ we have that $\phi_{\mathcal{G}}$ induces an action on 
$$\mathrm{Gr}_{q^r}(s,\Fm^k)\coloneqq\{ \mathcal{V} \subseteq \Fm^k \st {\mathcal{V} \mbox{ is an } \F_{q^r}\mbox{-subspace and }} \dim_{\F_{q^r}}(\mathcal{V})=s\}. $$
However, in this case, the action of $\mathcal{G}$ is not always faithful, since its kernel is given by $\mathcal{G}\cap \D_{q^r}$, where $\D_{q^r}=\{\alpha I_k \st \alpha \in \F_{q^r}^*\}$. 
However, if for any $r$ dividing $m$ we restrict to study equivalence classes of the Grassmannian $\mathrm{Gr}_{q^r}(s,\Fm^k)$ according to Definition \ref{def:equiv_systems}, then it is clear that also $\cG\cap \D_{q^m}$ acts trivially. Thus, we have an induced action of the group $\overline{\cG}\coloneqq \mathcal{G}/(\mathcal{G}\cap \D_{q^m})$ on the set $\mathfrak U[s,k,d]_{q^m/q^r}$, which we denote by  
$\phi_{\mathcal{G}}^{r,s}$.

\noindent Of particular interest for us is the action $\phi_{\mathcal{G}^\top}^{m,k-1}$ that 
$$\mathcal{G}^\top\coloneqq \left\{A^\top \st A \in \mathcal G \right\}$$
induces on $\mathrm{Gr}_{q^m}(k-1,\Fm^k)$, that is the set of $\Fm$-hyperplanes of $\Fm^k$. For brevity, we  will denote it by $\hat{\phi}_{\mathcal{G}}$.

From now on we will say that a group $\mathcal{G}\le\GL(k,q^m)$ is \textbf{transitive} whenever the action $\phi_{\mathcal{G}}^{m,1}$ is  transitive. The following result is straightforward.

\begin{lemma}\label{lem:phi_transitive}
If $\mathcal{G}$ is transitive, then also $\hat{\phi}_{\mathcal{G}}$ is transitive. 
\end{lemma}

\begin{proof}
By duality, it is equivalent to show that $\mathcal{G}^\top$ acts transitively on the one-dimensional $\F_{q^m}$-subspaces of $\Fm^k$, or, in other words, that the action $\phi_{\mathcal{G}^\top}^{m,1}$ is transitive. Let $X\coloneqq \mathrm{Gr}_{q^m}(1,\Fm^k)$, and for any matrix $A\in\mathcal G$ denote by $X^A$ to be the set of fixed points of $X$ under $A$, that is 
$$ X^A:=\left\{V\in X \st \phi_{\mathcal G}^{m,1}(A,V)=V \right\}.$$ 
However, any element in $X$ is of the form $\langle v\rangle_{\Fm}$ for some $v \in \Fm^k\setminus \{0\}$, and it holds that 
\begin{align*} X^A&=\{\langle v\rangle_{\Fm} \st \langle vA\rangle_{\Fm}=\langle v\rangle_{\Fm}  \} \\
&= \{\langle v\rangle_{\Fm} \st vA=\lambda v, \mbox{ for some } \lambda \in \Fm^*\} \\
&=\Big\{ \langle v\rangle_{\Fm} \in  X \st v \in \bigcup_{\lambda \in \sigma_A} A_{\lambda}\Big\},
\end{align*}
where $\sigma_A$ denotes the spectrum of $A$, and $A_\lambda$ is the eigenspace of $A$ with respect to the eigenvalue $\lambda$.
By Burnside's Lemma, one has that the number of orbits of the actions $\phi_{\mathcal{G}}^{m,1}$ and $\phi_{\mathcal{G}^\top}^{m,1}$ are respectively
\begin{align*}|X/\mathcal G|&=  \frac{1}{|\mathcal G|}\sum_{A\in\mathcal G}|X^A| \\
|X/\mathcal G^\top|&=  \frac{1}{|\mathcal G|}\sum_{A\in\mathcal G}|X^{A^\top}|
\end{align*}
Since $A$ and $A^\top$ have the same spectrum and isomorphic eigenspaces, we conclude that
$$ 1=|X/\mathcal G|=|X/\mathcal G^\top|,$$
and therefore $\phi_{\mathcal{G}^\top}^{m,1}$ is transitive.
\end{proof}

At this point we can state the main result of this section, which allows to construct several constant rank-profile codes.

\begin{theorem}\label{thm:orbit_construction}
Let $r$ be a divisor of $m$ and let $t$ be a positive integer such that $n\coloneqq rt \leq km$.  Let $\mathcal{G}\le\GL(k,q^m)$ be a transitive subgroup and let $\cO=(\cU_1,\ldots,\cU_t)$ be an orbit of the action $\phi_{\mathcal{G}}^{r,s}$. Then every code $\C\in\Phi([\cO])$ is a constant rank-profile code. 
\end{theorem}

\begin{proof}
 Let $\mathcal{H}_1=v_1^\perp$, $\mathcal{H}_2=v_2^\perp \subseteq \Fm^k$ be two $\Fm$-hyperplanes. It is enough to show that the dimension-list $\nu_{\cO}(H_1)$ can be obtained by permuting the dimension-list $\nu_{\cO}(H_2)$. Since $\mathcal{G}$ is transitive, by Lemma \ref{lem:phi_transitive} also $\widehat{\phi}_{\mathcal{G}}$ is transitive, hence there exists
 $A\in \mathcal{G}$  such that $(v_1A^\top)^\perp=v_2^\perp$. Then, we have
 \begin{align*}
 \nu_{\cO}(H_2)&=\nu_{\cO}((v_1A^\top)^\perp)=(\dim_{\Fq}((v_1A^\top)^\perp\cap \mU_1),\ldots, \dim_{\Fq}((v_1A^\top)^\perp\cap \mU_t)) \\
 &=(\dim_{\Fq}(v_1^\perp\cap \phi_{\mathcal{G}}^{r,s}(A,\mU_1)),\ldots, \dim_{\Fq}(v_1^\perp\cap \phi_{\mathcal{G}}^{r,s}(A,\mU_t)))\\ 
 &=\nu_{\phi_{\mathcal{G}}^{r,s}(A,\cO)}(H_1).
 \end{align*}
Since $\cO$ is an orbit of the action $\phi_{\mathcal{G}}^{r,s}$, then $\phi_{\mathcal{G}}^{r,s}(A,\cO)$ is just a permutation of $\cO$, and hence $\nu_{\cO}(H_2)$ is a permutation of $ \nu_{\cO}(H_1)$. 
\end{proof}

In order to use Theorem \ref{thm:orbit_construction} for concrete constructions of constant rank-profile codes, we need to find concrete examples of transitive subgroups $\mathcal{G}\le \GL(k,q^m)$.

\begin{remark}
It is worth to note that the concatenation of two constant rank-profile codes is a rank-profile code as well.
\end{remark}

In the following remark we point out that the construction of one-weight sum-rank metric code in Theorem \ref{th:RS1w} cannot be obtained as a concatenation of orbital constructions.

\begin{remark}
Consider the $\fq$-subspaces associated with the doubly-extended linearized Reed-Solomon code of dimension two: $\mathcal{U}_i=\{(y,a_i \sigma(y))\st y \in \F_{q^m}\}$, for any $i\in[q-1]$ where the $a_i$'s have pairwise distinct norm over $\fq$, $\mathcal{U}_0=\langle(1,0)\rangle_{\fq} $ and $\mathcal{U}_{\infty}=\langle(0,1)\rangle_{\fq}$. 
Suppose that these subspaces have been obtained as the union of orbit under the action of some transitive groups. Since $\mathcal{U}_0$ and $\mathcal{U}_{\infty}$ are the only $\fq$-subspaces of dimension one, then $\{\mathcal{U}_0,\mathcal{U}_{\infty}\}$ is an orbit under a transitive subgroup of $\mathrm{GL}(2,q^m)$, which is not possible.
\end{remark}

\subsection{The simplex codes in the sum-rank metric}

In this section we define the concept of an $n$-simplex code for the sum-rank metric. This code is an $[(n,\ldots,n),k,d]_{q^m/q}$. We will show that it is constant rank-profile and its parameters are fully determined by the weight distribution of a special rank-metric code associated to it. Furthermore, we discuss the equivalence classes of $n$-simplex codes. 

\begin{definition}
 A \textbf{Singer subgroup} $\mathcal{G}\le \GL(k,q^m)$ is a cyclic subgroup of order ${q^{km}-1}$. \end{definition}
One may construct a Singer subgroup $\mathcal{G}\le \GL(k,q^m)$ by taking the multiplicative subgroup of the finite field $\F_{q^{km}}$. More specifically, we can take $\cH_p\coloneqq\{M_p^i \st 0\le i \le q^{km}-2\}$, where $M$ is the companion matrix of an irreducible polynomial $p$ of degree $k$ over $\Fm$. 
Furthermore, it is well-known that any Singer subgroup $\mathcal{G}$ is conjugate to $\cH_p$, that is
$$ \mathcal{G}=B \cH_p B^{-1}=\left\{BM_p^iB^{-1} \st 0 \le i \le q^{km}-2\right\}, $$
for some $B \in \GL(k,q^m)$.

\begin{definition}\label{def:simplex}
 Let $\mU$ be an $n$-dimensional $\Fq$ subspace of $\Fm^k$ and let $\mathcal{G}$ be a Singer subgroup of $\GL(k,q^m)$. An \textbf{$n$-simplex code} $\C_{\mathcal{G},\mU}$ is any $[(n,\ldots,n),k]_{q^m/q}$ code belonging to $\Phi([\cO])$, where $\cO=(\phi_{\mathcal{G}}^{1,n}(A,\mU))_{A \in \overline{\mathcal{G}}}$. If $\mathcal{G}=\cH_p$, then we say that the $n$-simplex code is \textbf{standard}.
\end{definition}

\begin{example}\label{ex:simplex}
Let $k=m=q=2$ and $n=3$. Let $\F_4=\F_2(\alpha)$, where $\alpha^2+\alpha+1=0$ and let us fix the $3$-dimensional $\Fq$-subspace 
$$\mU\coloneqq \left\{\begin{pmatrix} (a+b)\alpha+b, & a+c\alpha\end{pmatrix} \st a,b,c\in \F_2 \right\}\subseteq \F_4^2.$$
Let us consider the field extension $\F_{16}=\F_2(\beta)$, where $\beta^4+\beta+1=0$. One can compute the minimal polynomial of $\beta$ over $\F_2(\alpha)$, obtaining $p(x)=x^2+x+\alpha^2$.
The Singer subgroup $\cH_p$ is therefore given by 
$\cH_p\coloneqq \left\{ M_p^i \st 0\leq i \leq 14 \right\},$
where
$$M_p\coloneqq \begin{pmatrix} 0 & \alpha^2  \\ 1 & 1\end{pmatrix},$$
while $\overline{\cH}_p\cong \left\{ M_p^i \st 0\leq i \leq 4 \right\}.$
An $\Fq$-basis for $\mU$ is given by $(\alpha,1), (\alpha^2,0),(0,\alpha)$. Thus, the $3$-simplex code $\C_{\cH_p,\mU}$ is a $[(3,\ldots,3),2]_{4/2}$ code with generator matrix 
$$G=( \,A \, \mid \, (M_p)^\top A\, \mid \, \ldots \, \mid \, (M_p^\top)^{4}A\,), $$
where 
$$A\coloneqq \begin{pmatrix}\alpha & \alpha^2 & 0 \\ 1 & 0 & \alpha \end{pmatrix}.$$
Thus, we have
$$ G=\left( \begin{array}{ccc|ccc|ccc|ccc|ccc}
\alpha & \alpha^2 & 0 & 1 & 0 & \alpha & 0 & \alpha & \alpha & \alpha^2 & \alpha & \alpha^2 & \alpha^2 & \alpha^2 & \alpha \\
1&0 & \alpha & 0 & \alpha & \alpha & \alpha^2 & \alpha & \alpha^2 & \alpha^2 & \alpha^2 & \alpha & 1 & \alpha & 0
\end{array}\right),$$
the only sum-rank weight of this code is $9$ and the only rank-profile is $(2,2,2,2,1)$. This is an example of a constant rank-profile code whose parameters are those of Example \ref{exa:constantrankprofile} with $q=2$ and $t=(q^2+1)$ and $t'=1$.
\end{example}

We have just seen in Example \ref{ex:simplex} a construction of a $3$-simplex code, which has constant weight $9$ and constant rank-profile $(2,2,2,2,1)$. Thus, it is natural to ask what can be said in general about the only rank-profile and the only sum-rank weight of an $n$-simplex code in the sum-rank metric. The following result provides an answer to this question, linking the rank-profile and the constant weight to the weight distribution and to the total weight of a particular rank-metric code.

In order to do that we introduce some notations.  For a given $\Fq$-subspace $\mU \subseteq \Fm^k$ such that $s\coloneqq \dim_{\Fm}\langle \mathcal U\rangle_{\Fm}$, define
$\widetilde{\mU}$ to be an $\Fm$-linear embedding of $\mU$ in $\Fm^s$, that is an $[n,s]_{q^m/q}$ system.

Furthermore, for an element $v =(v_1, \ldots, v_t)\in \Fm^{\bfn}$, we define its \textbf{block-Hamming support}, to be the set 
$$\supp_{\bfn}^{\HH}(v)\coloneqq \{i \in [t] \st v_i \neq 0\}.$$

\begin{theorem}
Let $\mU$ be an $n$-dimensional $\Fq$-subspace of $\Fm^k$, let $\mathcal{G}$ be a Singer subgroup of $\GL(k,q^m)$, and let $t\coloneqq \frac{q^{km}-1}{q^m-1}$. Let $\widetilde{\mU}$ be the $[n,s]_{q^m/q}$ system deriving from $\mU$, and let $\widetilde{\C}\in \Phi([\widetilde{\mU}])$ be any associated rank-metric code.
Then, for every nonzero 
$c=(c_1, \ldots, c_t)\in \C_{\mathcal{G},\mU}$,
the map 
$$\begin{array}{rccc} \iota : &\supp_{\bfn}^{\HH}(c) & \longrightarrow & \PG(\widetilde{\mathcal{C}},\Fm) \\
&i & \longmapsto & \langle c_i\rangle_{\Fm}
\end{array}$$
is $q^{(k-s)m}$-to-$1$.
In particular,
$\supp_\bfn(c)$ is obtained by rearranging $q^{(k-s)m}$ copies of the vector $(\supp(\tilde{c}))_{\tilde c \in \PG(\widetilde{\C},\Fm)}$
together with $\binom{k-s}{1}_{q^m}$ occurrences of the trivial space $\{0\}$.

Thus, the nonzero weight of $\C_{\mathcal{G},\mU}$ is
\[ w=\frac{q^{(k-s)m}}{q^m-1} \sum_{c \in \widetilde{\C}}\ww(c). \]
\end{theorem}

\begin{proof}
Let $c=(c_1, \ldots, c_t)\in \C_{\mathcal{G},\mU}$ and let $\{u_1,\ldots,u_n\}$ be an $\F_q$-basis of $\mathcal{U}$ and let $G \in \F_{q^m}^{k\times n}$ be the matrix whose columns are $u_1,\ldots,u_n$.
Suppose that $\mathcal{\overline{G}}=\{A_1=id,A_2,\ldots,A_{t}\}$. It follows that a generator matrix of $\C_{\mathcal{G},\mU}$ is $(G|A_2^\top G|\ldots|A_t^\top G)$, since $\{u_1A_i,\ldots,u_nA_i\}$ is an $\F_q$-basis of $\phi_{\mathcal{G}}^{1,n}(A_i,\mathcal{U})$. Moreover, note that the span over $\F_{q^m}$ of the rows of $G$ and the span over $\F_{q^m}$ of the rows $A_i^\top G$ coincide. 
Let $u'_1,\ldots,u'_s \in \mathcal{U}$ and $v'_1,\ldots,v'_s \in \tilde{\mathcal{U}}$ be two $\F_{q^m}$-basis of $\langle\mathcal{U}\rangle_{\F_{q^m}}$ and $\langle\tilde{\mathcal{U}}\rangle_{\F_{q^m}}$, respectively.
Let $B\in \F_{q^m}^{k\times s}$ with $\rk(B)=s$ such that $u'_iB=v'_i$ for any $i\in [s]$. It follows that $c_i \in \mathrm{rowspan}_{\F_{q^m}}(A_i^\top G)=\mathrm{rowspan}_{\F_{q^m}}(G)=\mathrm{rowspan}_{\F_{q^m}}(\tilde{G})$, where the columns of $\tilde{G}$ are the $u_iB$'s and hence it is a generator matrix of $\tilde{\mathcal{C}}$. Then the application $\iota$ is well-defined.
Since $\mathcal{G}$ is transitive and has order $q^{km}-1$ there is a bijection between the element of $\mathcal{G}$ and the nonzero vector of $\F_{q^m}^k$ and hence a bijection between $\overline{\mathcal{G}}$ and the points of $\PG(k-1,q^m)$.
Since $\dim_{\F_{q^m}}(\ker_{\F_{q^m}}(G))=k-s$, it follows that $\iota$ is $q^{(k-s)m}$-to-$1$. Moreover, a codeword $(c_1,\ldots,c_t)=(vG,vA_2^\top G,\ldots,vA_t^\top G)$ has exactly $\frac{q^{(k-s)m}-1}{q^m-1}=\binom{k-s}{1}_{q^m}$ zero entries since $vA_i^T$ covers the non zero vectors of $\F_{q^m}^{k}$, up to a scalar in $\F_{q^m}^*$. Since $\iota$ is $q^{(k-s)m}$-to-$1$ its image has size $(\frac{q^{km}-1}{q^m-1}-\frac{q^{(k-s)m}-1}{q^m-1})\cdot \frac{1}{q^{(k-s)m}}=\lvert \PG(\widetilde{\C},\Fm) \rvert$. It follows that $\supp_\bfn(c)$ is obtained by rearranging $\binom{k-s}{1}_{q^m}$ occurrences of the trivial space $\{0\}$ together with $(\supp(\tilde{c}))_{\tilde c \in \PG(\widetilde{\C},\Fm)}$ counted exactly $q^{(k-s)m}$ times.
\end{proof}

Back to Example \ref{ex:simplex}, then one can see that  $\widetilde{\C}$ is the code whose generator matrix is
$$ A\coloneqq \begin{pmatrix}\alpha & \alpha^2 & 0 \\ 1 & 0 & \alpha \end{pmatrix},$$
and the set of (representatives of) the projective codewords of $\widetilde{\C}$ is 
$$\left\{(\alpha,\alpha^2,0),(1, 0, \alpha), (\alpha^2,\alpha^2,\alpha), (0,\alpha^2,\alpha^2),(1,\alpha^2,1)\right\}. $$

\bigskip

We now observe that we can just restrict to study the $n$-simplex code obtained from the Singer subgroup $\cH_p$. This is due to the following result, showing that all the $n$-simplex codes are equivalent to a standard $n$-simplex code. This is a consequence of the fact that all the Singer subgroups are conjugated.

\begin{proposition}
 Let $B\in\GL(k,q^m)$ and $\mathcal{G}=B\cH_p B^{-1}$. Then $\C_{\mathcal{G},\mU}$ is equivalent to $\C_{\cH_p,\phi_{\cG}^{1,n}(B, \mU)}$. 
\end{proposition}

\begin{proof}
We have that $\C_{\mathcal{G},\mU}$ and $\C_{\cH_p,\phi_{\cG}^{1,n}(B, \mU)}$ are equivalent if and only any of their associated $\Fmnk$ systems are equivalent. By definition of $n$-simplex codes we can take $(\phi_{\cG}^{1,n}(A,\mU))_{A \in \overline{\cG}}$ as  $\Fmnk$ system associated to $\C_{\mathcal{G},\mU}$. Thus,
\begin{align*}(\phi_{\cG}^{1,n}(A,\mU))_{A \in \overline{\cG}} &=(\phi_{\cG}^{1,n}(A,\mU))_{A \in B\overline{\cH}_pB^{-1}}=(\phi_{\cG}^{1,n}(BAB^{-1},\mU))_{A \in \overline{\cH}_p}\\
&=(\phi_{\cG}^{1,n}(AB^{-1},\phi_{\cG}^{1,n}(B, \mU)))_{A \in \overline{\cH}_p}=(\phi_{\cG}^{1,n}(A,\phi_{\cG}^{1,n}(B, \mU)))_{A \in \overline{\cH}_p}\cdot B^{-1},
\end{align*}
which is clearly equivalent to any $\Fmnk$ system associated to $\C_{\cH_p,\phi_{\cG}^{1,n}(B, \mU)}$.
\end{proof}

As a consequence we obtain the following result on the number of inequivalent $n$-simplex codes.

\begin{corollary}\label{cor:number_inequivalent_simplex}
 The number of inequivalent $n$-simplex codes is at most the number of orbits of the action $\phi_{\cH_p}^{1,n}$.
\end{corollary}

\begin{remark}\label{rem:simplex_special_cases}
The definition of $n$-simplex code generalizes those of simplex codes in the Hamming and in the rank metric. For the Hamming metric, we have $n=1$, and  any transitive group $\mathcal{G}$ clearly acts transitively also on the $1$-dimensional subspaces of $\Fm^k$ (that is $\PG(k-1,q^m)$). Thus, the only possibility for an orbital construction proposed in Theorem \ref{thm:orbit_construction}  is that the orbit is the full $\PG(k-1,q^m)$, and  therefore the only code obtained in this way is the $[\frac{q^{km}-1}{q^m-1},k]_{q^m}$ simplex code. 
If we try to recover a one-weight code in the rank metric from Theorem \ref{thm:orbit_construction}, we need instead that $t=1$, that is the orbit of $\phi_{\mathcal{G}}^{r,s}$ consists of only one subspace $\mU$. Since $\mathcal{G}$ is transitive, this is only possible when $\mU=\Fm^k$, leading to the simplex code in the rank metric defined in \cite{Randrianarisoa2020ageometric,alfarano2021linear}.

However, in contrast with the Hamming and the rank metric, in the sum-rank metric one-weight codes are not unique up to equivalence and repetition. As shown in Corollary \ref{cor:number_inequivalent_simplex} we may have many distinct inequivalent code with the same parameters. The number of such equivalence classes coincides with the number of orbits of the action $\phi_{\cH_p}^{1,n}$. These orbits have been studied in various other contexts; see e.g. \cite{drudge2002orbits}.
\end{remark}

\begin{remark}
Simplex codes in the sum-rank metric have already been introduced in \cite{martinezpenas2021hamming}, but with a different meaning. Indeed, the definition given there generalizes the property of simplex codes in the Hamming metric of being the dual of the only nontrivial perfect codes with minimum distance $3$. Furthermore, those generalizations of simplex codes are not constant weight and it is not shown any systematic construction nor even their existence. What we are instead proposing here is to generalize simplex codes in the sum-rank metric with the aim of keeping the property of having only one nonzero weight. This leads to  codes sharing also the geometric and group theoretical aspects with classical simplex codes. For these reasons, we believe that our generalization of simplex codes may have a deep mathematical meaning and it certainly deserves attention. 
\end{remark}

\begin{remark}
Consider $k$ as an odd positive integer and $\mathcal{U}=\fq^k\subseteq \F_{q^2}^k$ and let $\theta$ be the Singer cycle induced by a primitive element in $\F_{q^{2k}}$. The orbit $\mathcal{O}$ of $\mathcal{U}$ under the action of the group $\mathcal{G}$ generated by $\theta^{\frac{q^k-1}{q-1}}$ has size $\frac{q^k+1}{q+1}$, \cite{hirschfeld1998projective}.
From a geometric point of view, the linear sets associated to the $\fq$-subspaces of $\mathcal{O}$ form a partition of $\PG(k-1,q^2)$ in Baer subgeometries, which is known as \textbf{classical Baer subgeometry partition}, see \cite{baker2000baer}.
In \cite{santonastaso2021subspace}, it has been proved that a Baer subgeometry partition yields to a one-weight sum-rank metric code.
Even if the group $\mathcal{G}$ is not transitive and hence the code associated with $\mathcal{O}$ is not an $n$-simplex, the orbit $\mathcal{O}$ still defines a one-weight code.
See \cite{ebert1998partitioning} for other Baer subgeometry partitions which are not classical.
\end{remark}

\section{Non-homogeneous case and multi-linear sets}\label{sec:non-hom}

In the above section, we have investigated  one-weight sum-rank metric codes in which all the blocks have the same length.
In this section we deal with one-weight codes of dimension two with possibly different length of the blocks; for this reason we call them \textbf{non-homogeneous}. By Theorem \ref{th:connection} and by the definition of weight of a point with respect to a linear set, the property of being a one-weight code may be rephrased in geometric terms.

Let $\mathcal{U}=(\mathcal{U}_1,\ldots,\mathcal{U}_t)$, where $\mathcal{U}_i$ is an $\fq$-subspace of $\F_{q^m}^k$ for every $i$. We define the \textbf{$\fq$-multi-linear set} $L_{\mathcal{U}}$ as the set of points in $L_{\mathcal{U}_1}\cup\ldots\cup L_{\mathcal{U}_t}$ and the \textbf{multi-weight} $w_{L_{\mathcal{U}}}(P)$ of a point $P \in \PG(k-1,q^m)$ with respect to $L_{\mathcal{U}}$ is 
\[ w_{L_{\mathcal{U}}}(P)\coloneqq\sum_{i=1}^t w_{L_{\mathcal{U}_i}}(P). \]
We say that an $\fq$-multi-linear set $L_{\mathcal{U}}$ is \textbf{scattered} if 
\[ w_{L_{\mathcal{U}}}(P)\leq 1, \]
for every point $P\in \PG(1,q^m)$.
If $L_{\mathcal{U}}$ is a scattered $\fq$-multi-linear set then $L_{\mathcal{U}_i}$'s are pairwise disjoint and they also turn out to be scattered.

Let $\mathcal{C}$ be a non-degenerate $[\mathbf{n},2,d]_{q^m/q}$ code. Let $\mathcal{U}=(\mathcal{U}_1,\ldots,\mathcal{U}_t) \in \Psi([\mathcal{C}])$. Then $\mathcal{C}$ is one-weight if and only if
\begin{equation} \label{eq:geometriconeweight2}
 \ww_{L_{\mathcal{U}}}(P)=N-d, \mbox{      for each }P \in \PG(1,q^m),
\end{equation}

that is the linear sets associated with  a non-degenerate one-weight code have to cover all the projective line and each point must have the same multi-weight.

\subsection{Two-dimensional one-weight MSRD codes}

In this section we characterize those one-weight codes which are also MSRD codes and we provide bounds on $t$ in such a case. The first part is a consequence of \eqref{eq:geometriconeweight2} and Theorem \ref{th:Singletonbound}.

\begin{corollary} \label{prop:multiplicitymsrd}
Let $\mathcal{U}=(\mathcal{U}_1,\ldots,\mathcal{U}_t)$ be a $[\mathbf{n},2,d]_{q^m/q}$ system. 
Then $\Phi([\mathcal{U}])$ is a class of one-weight MSRD code if and only if $L_{\mathcal{U}}$ is a scattered $\fq$-multi-linear set and $L_{\mathcal{U}}=\PG(1,q^m)$.
\end{corollary}

The above corollary extends the already known connection between MRD codes and scattered linear sets (see e.g.\ \cite{sheekey2016new,polverino2020connections}).
Corollary \ref{prop:multiplicitymsrd} yields to a lower and an upper bound on the value $t$ when the code considered is MSRD.

\begin{corollary}
Let $\mathcal{\mathcal{\mathcal{U}}}=(\mathcal{U}_1,\ldots,\mathcal{U}_t)$ be a $[\mathbf{n},2,d]_{q^m/q}$ system and let $\C \in \Phi([\mathcal{U}])$. Then, $\C$ is a one-weight MSRD code if and only if
\begin{equation}\label{eq:scattered_oneweight}q^m+1 = |L_{\mathcal{U}}|= \sum_{i=1}^t \frac{q^{n_i}-1}{q-1}.\end{equation}
In particular, in this case we have
\begin{equation}\label{eq:boundt}
q+1 \leq t \leq q^m+1
\end{equation}
and $t\equiv 1 \pmod q$.
\end{corollary}

\begin{proof}
By Corollary \ref{prop:multiplicitymsrd}, we have that $w_{L_{\mathcal{U}}}(P) = 1$, for each $P \in \PG(1,q^m)$, that is $L_{\mathcal{U}_i}$ is scattered and the $L_{\mathcal{U}_i}$'s are pairwise disjoint. Then by \eqref{eq:boundscatt} $ \dim_{\F_q}(\mathcal{U}_i)=n_i \leq m$ and $\lvert L_{\mathcal{U}_i} \rvert =\frac{q^{n_i}-1}{q-1}$. Since $\Phi([\mathcal{U}])$ is a class of one-weight MSRD  code, again by Corollary \ref{prop:multiplicitymsrd}, we have that 
\[q^m+1 = |L_{\mathcal{U}}|= \sum_{i=1}^t \,\,\lvert L_{\mathcal{U}_i} \rvert.
\]
Then by \eqref{eq:boundsizels} we have 
\[q^m+1=|L_{\mathcal{U}}|=\sum_{i=1}^t \frac{q^{n_i}-1}{q-1} \leq t \frac{q^m-1}{q-1},\]
since $L_{\mathcal{U}_i}$ is scattered for every $i$. 
This implies that $t\geq q$. Now, if $t=q$, then we have
$$q^m+1=|L_{\mathcal{U}}|=\sum_{i=1}^t \frac{q^{n_i}-1}{q-1} \equiv 0 \pmod q,$$
which is a contradiction.
So $t$ must be different from $q$ and hence $t \geq q+1$. The upper bound follows from the fact that the number of points in $\PG(1,q^m)$ is $q^m+1$. Finally, one gets $t \equiv 1 \pmod q$ directly from \eqref{eq:scattered_oneweight}.
\end{proof}

\begin{remark}
We point out that the bounds on $t$ obtained in the above result are sharp.
Indeed, the value of $t$ in the system obtained from the doubly-extended linearized Reed-Solomon code $\mathcal{C}_{2,\mathrm{a},\boldsymbol{\beta}}^{\gamma,\delta}$ reaches the lower bound of \eqref{eq:boundt}. On the other hand, if $\PG(1,q^m)=\{P_1=\langle v_1 \rangle_{\F_{q^m}}, \ldots, P_{q^m+1}=\langle v_{q^m+1} \rangle_{\F_{q^m}}\}$, then $\mathcal{U}_1=\langle v_1\rangle_{\F_q},\ldots,\mathcal{U}_{q^m+1}=\langle v_{q^m+1} \rangle_{\F_q}$ and hence the value of $t$ of such a construction reaches the upper bound in \eqref{eq:boundt}. This latter construction is the projective system associated to the doubly-extended Reed-Solomon code, and it is well-known to be the unique -- up to equivalence -- $2$-dimensional one-weight code in the Hamming metric; see \cite{bonisoli1983every}.
\end{remark}

Thus, now we focus on one-weight MSRD codes whose number of blocks $t$ is $q+1$, that is extremal codes meeting the lower bound in \eqref{eq:boundt} with equality.

\begin{theorem}\label{th:boundq+1}
Let $\C$ be an $[(n_1,\ldots,n_{q+1}),2]_{q^m/q}$ one-weight MSRD code. Then, 
\begin{enumerate}
    \item either $n_1=\ldots=n_{q-1}=m$ and $n_q=n_{q+1}=1$,
    \item or $q=2$, $n_1=n_2=m-1$ and $n_3=2$.
\end{enumerate}
\end{theorem}

\begin{proof}
 Let $\C$ be an $[(n_1,\ldots,n_{q+1}),2]_{q^m/q}$ one-weight MSRD code and assume that $n_{q-1}\leq m-1$. Then by  \eqref{eq:scattered_oneweight}, we have
 $$ q^{m+1}-q^m+2q=\sum_{i=1}^t {q^{n_i}}\leq q^{m+1}-2q^m+3q^{m-1},$$
 from which we obtain that $q^m+2q\leq 3q^{m-1}$. It is easy to see that this is only possible for $q=2$, and hence for $q\geq 3$ the only possibility is the first case.
 
 Now, let us suppose that $q=2$.  In this case,  \eqref{eq:scattered_oneweight} becomes 
 $$ 2^{n_1}+2^{n_2}+2^{n_3}=2^m+4.$$
 If $n_1=m$ then clearly we must have $n_2=n_3=1$ and we get again the first case. Thus, assume that $n_1\leq m-1$. If $n_2\leq m-2$, then $2^{n_1}+2^{n_2}+2^{n_3}\leq 2^{m-1}+2^{m-2}+2^{m-2}$, which contradicts \eqref{eq:scattered_oneweight}. Then, the only possibility is that $n_1=n_2=m-1$ and $n_3=2$.
\end{proof}

Because of the existence of $2$-dimensional doubly-extended linearized Reed-Solomon codes, we know that case (1) in Theorem \ref{th:boundq+1} is always possible. We now see that also case (2) is possible, providing a construction for every $m\geq 3$.

\begin{theorem}\label{thm:oneweight_MSRD_sporadic}
Let $m$ be a positive integer with $m\geq 3$ and let $\mathcal{H}$ be an $(m-1)$-dimensional $\F_2$-subspace of $\F_{2^m}$.
For $\delta \notin \mathcal{H}$, define \begin{align*}
    \mathcal{X}&=\{(x,x^2)\st x \in \mathcal{H}\},\\
    \mathcal{Y}_{\delta}&=\{(x,x^2+\delta x)\st x \in \mathcal{H}\}, \\
    \mathcal{Z}_{\delta}&=\langle (1,0),(0,\delta) \rangle_{\fq},
\end{align*} 
and $\mathcal{U}=(\mathcal{X},\mathcal{Y}_{\delta},\mathcal{Z}_{\delta})$. Then, any code in $\Phi([\mathcal{U}])$ is a one-weight MSRD $[(m-1,m-1,2),2,2m-1]_{2^m/2}$ code.
\end{theorem}
\begin{proof}
By Corollary \ref{prop:multiplicitymsrd}, we need to prove that $L_{\mathcal{U}}$ is a scattered $\fq$-multi-linear set and $L_{\mathcal{U}}=\PG(1,q^m)$.
The $\fq$-linear sets $L_{\mathcal{X}}$ and $L_{\mathcal{Y}_{\delta}}$ are two $\fq$-linear sets contained in $\PG(1,q^m)$ having rank $m-1$. 
The linear set $L_1=\{\langle (y,y^2)\rangle_{\F_{2^m}} \st y \in \F_{2^m}^* \}$ is scattered (see \cite{blokhuis2000scattered}) and, since $L_2=\{\langle (y,y^2+\delta y)\rangle_{\F_{2^m}} \st y \in \F_{2^m}^* \}$ is $\mathrm{PGL}(2,q^m)$-equivalent to $L_1$, $L_2$ is scattered as well.
Therefore, we have that $L_{\mathcal{X}}$ and $L_{\mathcal{Y}_{\delta}}$ are scattered, since $L_{\mathcal{X}}\subseteq L_1$ and $L_{\mathcal{Y}_{\delta}}\subseteq L_2$.
We note that
\[ L_{\mathcal{X}}=\{\langle (1,y) \rangle_{\F_{2^m}} \st y \in \mathcal{H}^*\} \]
and 
\[ L_{\mathcal{Y}_{\delta}}=\{\langle (1,y+\delta) \rangle_{\F_{2^m}} \st y \in \mathcal{H}^*\}, \]
where $\mathcal{H}^*=\mathcal{H}\setminus\{0\}$, and hence $L_{\mathcal{X}}\cap L_{\mathcal{Y}_{\delta}}=\emptyset$.
Clearly, $L_{\mathcal{Z}_{\delta}}=\PG(1,q)$  is a scattered $\fq$-linear set and 
\[ L_{\mathcal{Z}_{\delta}}=\{ \langle (1,0) \rangle_{\F_{2^m}}, \langle (0,1) \rangle_{\F_{2^m}}, \langle (1,\delta) \rangle_{\F_{2^m}} \}. \]
Since $\delta$ is the only nonzero element in $\F_{2^m}\setminus(\mathcal{H}^*\cup (\delta+\mathcal{H}^*))$, we also have that $L_{\mathcal{Z}_{\delta}}\cap (L_{\mathcal{X}}\cup L_{\mathcal{Y}_{\delta}})=\emptyset$.
So, since $L_{\mathcal{X}},L_{\mathcal{Y}_{\delta}}$ and $L_{\mathcal{Z}_{\delta}}$ are pairwise disjoint and they are also scattered $\fq$-linear sets, it follows that $L_{\mathcal{U}}$ is a scattered $\fq$-multi-linear set such that $L_{\mathcal{U}}=\PG(1,q^m)$. 
\end{proof}

Any code in $\Phi([\mathcal{U}])$ will be called \textbf{$2$-fold linearized Reed-Solomon code}.

\begin{remark}
The $2$-fold linearized Reed-Solomon codes obtained in Theorem \ref{thm:oneweight_MSRD_sporadic} are very special and indeed they can only occur when $q=2$. To the best of our knowledge they cannot be recovered from known constructions. It is clear that they are not linearized Reed-Solomon codes \cite{martinezpenas2018skew}, since for $q=2$ they are constituted by only one block. Furthermore, twisted linearized Reed-Solomon codes coincide with linearized Reed-Solomon code when $q=2$ \cite{neri2021twisted}. Moreover, the constructions presented in \cite{martinez2020general} have all the same block length, and thus they do not coincide with $2$-fold linearized Reed-Solomon codes.

Moreover, these $2$-fold linearized Reed-Solomon codes cannot even be obtained by puncturing one of the known MSRD codes. This is easy to see since the puncturing of a sum-rank metric code $\C\in \Phi([\mU])$, with $\mU=(\mU_1,\ldots,\mU_t)$ corresponds to choosing a code $\C'\in \Phi([\mU'])$, where $\mU'=(\mU_1',\ldots, \mU_t')$ and $\mU'_i\subseteq \mU_i$ for each $i\in[t]$. However, since the $\Fq$-multi-linear set $L_{\mU'}$ with $\mU'=(\mathcal X, \mathcal X_\delta, \mathcal Z_\delta)$, covers the whole projective line $\PG(1,2^m)$, this cannot happen by Corollary \ref{prop:multiplicitymsrd}.
\end{remark}

\begin{example}
 Let $q=2$ and $m=4$, and let $\F_{16}=\F_2(\beta)$, where $\beta^4+\beta+1=0$. Let us choose $\mathcal H=\langle 1, \beta, \beta^2 \rangle_{\F_2}$, and $\delta=\beta^3$. We have that the sets $\mathcal X, \mathcal Y_\delta, \mathcal Z_\delta$ are 
 \begin{align*}\mathcal X&=\langle (1,1),(\beta,\beta^2), (\beta^2,\beta^4)\rangle_{\F_2}, \\
 \mathcal Y_\delta&=\langle (1,\beta^{14}),(\beta,\beta^{10}), (\beta^2,\beta^8)\rangle_{\F_2}, \\
 \mathcal Z_\delta &=\langle (1,0),(0,\beta^3) \rangle_{\F_2}.
 \end{align*}
 Thus, a $[(3,3,2),2,7]_{16/2}$ $2$-fold linearized Reed-Solomon code $\C\in \Phi([\mU])$ is the one generated by
 $$ G\coloneqq \left(\begin{array}{ccc|ccc|cc} 1 & \beta & \beta^2 & 1 & \beta & \beta^2 & 1 & 0 \\
 1 & \beta^2 & \beta^4 & \beta^{14} & \beta^{10} & \beta^8 & 0 & \beta^3\end{array}\right).$$
\end{example}

\begin{remark}
If $L_{\mathcal{U}}$ is a scattered $\F_2$-linear set of rank $m$ in $\PG(1,2^m)$, then $\mathcal{U}$ is $\GammaL (2,2^m)$-equivalent to the $\F_2$-subspace 
\[ \{ (y,y^2) \st y \in \F_{2^m} \}, \]
since $\mathcal{U}$ corresponds to a translation hyperoval, whose classification dates back to the paper by Payne \cite{payne1971complete}.
This implies that every one-weight MSRD $[(m,1,1),2,m+1]_{2^m/2}$ code is equivalent to a doubly-extended linearized Reed-Solomon code.
\end{remark}

\begin{remark}
Let $m=4$. If $L_{\mathcal{U}}$ is a scattered $\F_q$-linear set of rank $4$ in $\PG(1,q^4)$, then $\mathcal{U}$ is $\GammaL (2,q^4)$-equivalent either to the $\F_q$-subspace 
\[ \mathcal{V}=\{ (y,y^q) \st y \in \F_{q^4} \}, \]
or to
\[ \mathcal{W}=\{ (y^q,y^{q^2}+\delta y) \st y \in \F_{q^4} \}, \]
with $\delta \in \F_{q^4}$ with $\N_{q^4/q}(\delta)\notin \{0,1\}$, see \cite{bonoli2005fq,csajbok2018maximumPG} for the classification of scattered linear sets and then see \cite[Section 4]{csajbok2018classes} to obtain the classification in terms of subspaces.

This implies that every $\fq$-subspace of dimension $m$ of any system associated with a one-weight MSRD $[(m,\ldots,m,1,1),2,(q-1)m+1]_{q^4/q}$ code is $\mathrm{\Gamma L}(2,q^4)$-equivalent to either $\mathcal{V}$ or to $\mathcal{W}$.
\end{remark}

\bigskip

We now show that one can  construct $2$-dimensional one-weight MSRD codes, with $t$ blocks, also when $q+1<t<q^m+1$.

The example of one-weight doubly-extended linearized Reed-Solomon codes can be extended as follows.
First we recall the following result from \cite{neri2021twisted} for special sets of parameters.
To do this, we introduce the following notation. Let $\mathcal{G}$ be a group and let $g_1,\ldots,g_t \in \mathcal{G}$, we denote by $\langle g_1,\ldots,g_t\rangle$ the subgroup of $\mathcal{G}$ generated by $g_1,\ldots,g_t$.

\begin{theorem}[see \textnormal{\cite[Theorem 6.3]{neri2021twisted}}]\label{th:constrNeri}
Let $a=(a_1,\ldots,a_t) \in (\F_{q^m}^*)^t$ be such that $\N_{q^m/q}(a_i) \neq \N_{q^m/q}(a_j)$ if $a_i \neq a_j$. 
Suppose that there exists $\eta \in \F_{q^m}$ such that $\N_{q^m/q}(\eta) \notin \langle \N_{q^m/q}(a_1), \ldots, \N_{q^m/q}(a_t) \rangle\subseteq \F_{q}^*$. 
Let $\gamma=(\gamma_1,\ldots,\gamma_m)$ be an ordered $\F_q$-basis of $\F_{q^m}$.
Define
\[ T_\eta\coloneqq\{f=x_0 + x_1 x +x_0 \eta x^2 \st x_0,x_1 \in \F_{q^m}\} \subseteq \F_{q^m}[x;\sigma].\]
Then the code 
\[
\mathcal{H}_{a,\gamma}^{\eta}\coloneqq\{ \mathrm{ev}_{a,\gamma}(f) \st f \in T_\eta\} \subseteq \F_{q^m}^{\mathbf{m}}
\]
is an $[\mathbf{m},2,tm-1]_{q^m/q}$ MSRD code, with $\mathbf{m}=(\underbrace{m,\ldots,m}_{t \mbox{ \scriptsize{times}}})$. 
\end{theorem}

The code $\mathcal{H}_{a,\gamma}^{\eta}$ was introduced in \cite[Definition 6.2]{neri2021twisted} and it is called \textbf{twisted linearized Reed-Solomon code}.

\begin{lemma}\label{lem:Neriex}
Let $\{a_1,\ldots,a_t\} \subseteq \F_{q^m}^*$ be such that $\N_{q^m/q}(a_i) \neq \N_{q^m/q}(a_j)$ if $a_i \neq a_j$. Suppose that there exists $\eta \in \F_{q^m}$ such that $\N_{q^m/q}(\eta) \notin \langle  \N_{q^m/q}(a_1), \ldots, \N_{q^m/q}(a_t) \rangle$. 
Let
\[\mathcal{U}_i=\{(y+\sigma^2(y) \eta N_2(a_i),a_i \sigma(y))\st  y \in \F_{q^m}\} \subseteq \F_{q^m}^2.\] 
Let $\mathcal{U}=(\mathcal{U}_1,\ldots,\mathcal{U}_t)$, then $L_{\mathcal{U}}$ is a scattered $\F_q$-multi-linear set in $\mathrm{PG}(1,q^m)$.
\end{lemma}
\begin{proof}
By Theorem \ref{th:connection} and Theorem \ref{th:constrNeri}, it follows that $ w_{L_{\mathcal{U}}}(\langle v \rangle_{\F_{q^m}}) \leq 1$,
for every point $\langle v\rangle_{\F_{q^m}}\in \PG(1,q^m)$.
This means that each point $P \in \PG(1,q^m)$ belongs to at most one of the linear sets $L_{\mathcal{U}_1},\ldots,L_{\mathcal{U}_t}$ and $w_{L_{\mathcal{U}_i}}(P)\leq 1$.
Hence the assertion is proved.
\end{proof}

\begin{remark}
Choosing $\eta=0$, we return to the case of Lemma \ref{lem:disjointpseudo}.
\end{remark}

\begin{remark}
Note that for the twisted linearized Reed-Solomon codes we have $t \leq q-1$. Moreover, the biggest proper subgroup of $\F_q^*$ has cardinality $\frac{q-1}{r}$, where $r$ is the smallest prime dividing $q-1$. Then, if we would construct a linearized Reed-Solomon codes with $\eta \ne 0$ we have $t \leq \frac{q-1}{r}$.  
\end{remark}

We can complete the set of such subspaces to get a one-weight MSRD code.
Indeed, let $t=(q-1)/r$, where $r$ is the smallest prime that divides $q-1$. Let $\mathcal{G}=\{g_1,\ldots,g_t\}$ be the subgroup of $\F_{q}^*$ of order $t$. Let $a_1,\ldots,a_t \in \F_{q^m}^*$ be elements with distinct norm such that $\N_{q^m/q}(a_i) \in \mathcal{G}$.  Let $\eta \in \F_{q^m}^*$ such that $\N_{q^m/q}(\eta) \notin \mathcal{G}$. Let $\mathcal{U}_i=\{(y+\sigma^2(y)\eta N_2(a_i),a_i \sigma(y))\st  y \in \F_{q^m}\} \subseteq \F_{q^m}^2$ and $\mathcal{U}=(\mathcal{U}_1,\ldots,\mathcal{U}_t)$. By Lemma \ref{lem:Neriex} the $L_{\mathcal{U}_i}$'s are pairwise disjoint and each of them has $\frac{q^m-1}{q-1}$ points. Let $z=q^m+1-\frac{q-1}{r}\frac{q^m-1}{q-1}$ and let $v_1,\ldots,v_z \in \F_{q^m}^2 \setminus \{(0,0)\}$ such that
\[\{P_1=\langle v_1 \rangle_{\F_{q^m}},\ldots,P_z=\langle v_z \rangle_{\F_{q^m}}\} =\mathrm{PG}(1,q^m) \setminus \{L_{\mathcal{U}} \}. \]  
Let $\mathcal{U}'_i=\langle v_i \rangle_{\fq}$ for every $i\in[z]$ and let $\mathcal{U}'=(\mathcal{U}'_1,\ldots,\mathcal{U}'_z)$. 
As a consequence of Lemma \ref{lem:Neriex} and of the definition of $\mathcal{U}_{i}$'s, we have that 
\[ w_{L_{\mathcal{U}}}(P) + w_{L_{\mathcal{U}'}}(P)=1,\] 
for every point $P \in \PG(1,q^m)$.
Therefore, we have the following result.

\begin{corollary}
Let $t=(q-1)/r$, where $r$ is the smallest prime that divides $q-1$. Let $\mathcal{G}=\{g_1,\ldots,g_t\}$ be the subgroup of $\F_{q}^*$ of order $t$. Let $a_1,\ldots,a_t \in \F_{q^m}^*$ be elements with distinct norm such that $\N_{q^m/q}(a_i) \in \mathcal{G}$.  Let $\eta \in \F_{q^m}^*$ such that $\N_{q^m/q}(\eta) \notin \mathcal{G}$. Let $\mathcal{U}_i=\{(y+\sigma^2(y)\eta N_2(a_i),a_i \sigma(y))\st  y \in \F_{q^m}\} \subseteq \F_{q^m}^2$.
Let $v_1,\ldots,v_z \in \F_{q^m}^2 \setminus \{(0,0)\}$ such that 
\[\{P_1=\langle v_1 \rangle_{\F_{q^m}},\ldots,P_z=\langle v_z \rangle_{\F_{q^m}}\} =\mathrm{PG}(1,q^m) \setminus \{ L_{\mathcal{U}}\}. \]  
Let $\mathcal{U}'_i=\langle v_i \rangle_{\fq}$ for every $i\in[z]$. 
Then the class of the $[\mathbf{n},2,mt+z-1]_{q^m/q}$ system $\mathcal{U}=(\mathcal{U}_1,\ldots,\mathcal{U}_t,\mathcal{U}'_1,\ldots,\mathcal{U}'_z)$, with $\mathbf{n}=(\underbrace{m,\ldots,m}_{t \mbox{ \scriptsize{times}} }, \underbrace{1,\ldots,1}_{z \mbox{ \scriptsize{times}}})$, defines a class $\Phi([\mathcal{U}])$ of non-degenerate one-weight MSRD $[\mathbf{n},2,mt+z-1]_{q^m/q}$ codes.
\end{corollary}

\subsection{Lifting construction}\label{sec:non-hom2}

We conclude this section by describing a procedure to construct one-weight sum-rank metric codes starting from a set of linear sets in $\PG(1,q^m)$. 

Let $\mathcal{U}_1,\ldots, \mathcal{U}_s$ be $\F_q$-subspaces in $\F_{q^m}^2$, let $\mathcal{U}=(\mathcal{U}_1,\ldots, \mathcal{U}_s)$ and consider $L_{\mathcal{U}}$ be the $\fq$-multi-linear set associated with $\mathcal{U}$.
Let 
\[ M=\max\left\{ w_{L_{\mathcal{U}}}(P) \st  P \in \PG(1,q^m)\right\}. \]

Denote by $\mathcal{M}(\mathcal{U})$ a vector of subspaces whose entries  (rearranged in decreasing order according to their dimension) are
\begin{itemize}
    \item $\mathcal{U}_1,\ldots,\mathcal{U}_s$;
    \item $c$ copies of $\langle v \rangle_{\F_{q^m}}$ and a $d$-dimensional subspace of $\langle v \rangle_{\F_{q^m}}$, for any $P=\langle v \rangle_{\F_{q^m}}\in \PG(1,q^m)$, where $M-w_{L_{\mathcal{U}}}(P)=c m+d$ with $c,d \in \mathbb{N}$ and $d<m$.
\end{itemize}

We call $\mathcal{M}({\mathcal{U}})$ the \textbf{lifting} of $\mathcal{U}$.
Clearly, $w_{L_{\mathcal{M}({\mathcal{U}})}}(P)=M$ for every $P\in \PG(1,q^m)$.

So, adding more copies of the points with weight less than the maximum value yields to the construction of one-weight sum-rank metric codes in the following way.

\begin{corollary}\label{cor:costantMLU}
Let $L_{\mathcal{U}_1},\ldots,L_{\mathcal{U}_s}$ be $s$ $\fq$-linear sets in $\PG(1,q^m)$ of rank $n_1,\ldots,n_s$, respectively, and let $\mathcal{U}=(\mathcal{U}_1,\ldots,\mathcal{U}_s)$.
Let 
\[ M=\max\left\{w_{L_{\mathcal{U}}}(P) \st  P \in \PG(1,q^m)\right\}. \] 
Let $\ell=\sum_{P \in L_{\mathcal{U}}} (M- w_{L_{\mathcal{U}}}(P))$ and $z=q^m+1-|L_{\mathcal{U}}|$.
The elements of an $\mathcal{M}(\mathcal{U})$ define an  $[\mathbf{n},2,d]_{q^m/q}$ system, with minimum distance $d=\sum_{i=1}^s n_i+(z -1)M + \ell$.

In particular, $\Phi([\mathcal{M}({\mathcal{U}})])$ is a class of non degenerate one-weight $[\mathbf{n},2,d]_{q^m/q}$-code.
\end{corollary}

Any code in $\Phi([\mathcal{M}({\mathcal{U}})])$ will be called the \textbf{lifted code of} $\mathcal{U}$.

\begin{remark}
We note that the above procedure can give constructions of one-weight codes which are not of the form in Section \ref{sec:orbital}.
Suppose that $L_{\mathcal{U}}\ne \PG(1,q^m)$ and $|L_{\mathcal{U}_i}|\geq 2$ for every $i$, let $P$ be a point in $\PG(1,q^m)$ not lying in $L_{\mathcal{U}}$ and let $Q=\langle w\rangle_{\F_{q^m}}$ be such that $w_{L_{\mathcal{U}}}(Q)=M$. 
Assume that $\mathcal{M}(\mathcal{U})$ is the concatenation of orbital constructions.
Then in the subspaces in the orbit of $\mathcal{U}_P$ under the action of any transitive group $\mathcal{G}$ should also cover at least one non-zero vectors in $\mathcal{U}_i\cap \langle w\rangle_{\F_{q^m}}$ (counted with multiplicity). This would imply that $\mathcal{M}(L_{\mathcal{U}})(Q)\geq M+1$, a contradiction to the above corollary.
\end{remark}

\begin{example}
Consider the $\fq$-linear set defined by $\mathcal{U}=\{(y,\mathrm{Tr}_{q^m/q}(y)) \st y \in \F_{q^m}\}$, where $\mathrm{Tr}_{q^m/q}$ denotes the trace function of $\F_{q^m}$ over $\fq$.
It is well-known that $L_{\mathcal{U}}$ is a club, that is there is only one point of $L_\mathcal{U}$ having weight $m-1$ (which is the point $\langle (1,0)\rangle_{\F_{q^m}}$ in this case) and all the remaining $q^{m-1}$ points have weight one; see e.g.\ \cite{de2016linear}. 
So, according to the above described procedure, the entries of $\mathcal{M}(\mathcal{U})$ (up to reordering) are:
\begin{itemize}
    \item $\mathcal{U}$;
    \item an $(m-2)$-dimension $\fq$-subspace of $\langle (1,\eta)\rangle_{\F_{q^m}}$ if $\langle (1,\eta)\rangle_{\F_{q^m}}\in L_\mathcal{U}$ and $\eta \ne 0$;
    \item an $(m-1)$-dimensional $\fq$-subspace of $\langle v\rangle_{\F_{q^m}}$ if $\langle v\rangle_{\F_{q^m}}\notin L_\mathcal{U}$.
\end{itemize}
By Corollary \ref{cor:costantMLU}, $\mathcal{M}(\mathcal{U})$ defines a non degenerate one-weight code $[\mathbf{n},2,d]_{q^m/q}$-code, with
\[
\mathbf{n}=(m,\underbrace{m-1,\ldots,m-1}_{q^m-q^{m-1} \mbox{ \scriptsize{times}}},\underbrace{m-2,\ldots,m-2}_{q^{m-1} \mbox{ \scriptsize{times}}}) \mbox{ and }d=(m-2)q^{m-1}+(m-1)(q^{m}-q^{m-1}) +1.
\]
\end{example}

\section{Conclusions and open problems}\label{sec:8}

In this paper we first introduce the notion of $[\mathbf{n},k,d]_{q^m/q}$ system, extending the ones of $q$-systems and projective systems already known. 
Then we have established a one-to-one correspondence between equivalence classes of $[\mathbf{n},k,d]_{q^m/q}$ systems and equivalence classes of linear $[\mathbf{n},k,d]_{q^m/q}$ sum-rank metric codes, which naturally extends the connections between equivalence classes of projective systems and linear Hamming-metric codes and the one between $q$-systems and linear rank metric codes.
Thank to this connection, we are able to introduce a new class of sum-rank metric codes which we call doubly-extended linearized Reed-Solomon code, since it is obtained by adding two blocks to a linearized Reed-Solomon code, similarly to the well-known class of doubly-extended Reed-Solomon codes in the Hamming metric. As one might expect, this is a family of MSRD codes and when they also have dimension $2$ then they are also one-weight codes.
Since every $[\mathbf{n},k,d]_{q^m/q}$ system also defines a projective system, to a linear sum-rank metric code we can associate a linear Hamming-metric code. As a consequence of the well-known result by Bonisoli, we obtain some restrictions on the parameters of a linear sum-rank metric code with constant rank-profile.
Moreover, transitive groups have been used to construct sum-rank metric codes with constant rank-profile, this also yields to a definition of simplex code by making use of Singer subgroups of $\mathrm{GL}(k,q^m)$. Such a definition of simplex code is different from the one presented by Martínez-Peñas, which was defined through the dual of nontrivial perfect codes with minimum distance $3$. 
Finally, we investigate $2$-dimensional linear one-weight sum-rank metric codes by means of the linear sets associated with an $[\mathbf{n},2,d]_{q^m/q}$ system, which we call $\fq$-multi-linear sets. This enables us to provide bounds on the number of blocks of a one-weight MSRD code and to characterize those that have the smallest number of blocks. We can also enlarge the construction of twisted linearized Reed-Solomon codes to get families of one-weight MSRD codes. We then describe a procedure to construct one-weight sum-rank metric codes starting from any $\fq$-multi-linear set contained in the projective line.

We conclude the paper with some questions/problems that we think could be of interest for the reader. 

\begin{open}
Is it possible to say more about \emph{extremal} one-weight sum-rank metric codes which are also MSRD codes? We have already shown in Theorem \ref{th:boundq+1} that for two-dimensional one-weight MSRD codes with exactly $q+1$ blocks, the block-lengths are uniquely determined for $q\geq 3$. Can this result be generalized to $k$-dimensional one-weight MSRD codes for $k\geq 3$?
\end{open}

\begin{open} The minimum distance of a $\mathbf{n}$-simplex code constructed in Section \ref{sec:orbital} through a subspace $\mathcal{U}$ and a Singer subgroup $\mathcal{G}$ of $\mathrm{GL}(k,q^m)$, depends on the total weight of the rank-metric codes associated with $\mathcal{U}$. It is therefore  of interest  to study the total weight of rank-metric codes.
\end{open}

\begin{open} 
The $2$-dimensional one-weight MSRD codes obtained from Theorem \ref{thm:oneweight_MSRD_sporadic}, called $2$-fold linearized Reed-Solomon, are special constructions  occurring only when $q=2$ and which, to the best of our knowledge, cannot be recovered from known constructions.
Can these examples be extended to higher dimensions?
\end{open}

\begin{open} In Section \ref{sec:non-hom2}, we describe a geometric construction for linear one-weight sum-rank metric codes of  dimension $2$ called the lifted code of $\mathcal{U}$. Can be this procedure extended in order to construct linear one-weight sum-rank metric codes of any dimension? 
\end{open}

\section*{Acknowledgements}

The authors are very grateful to Gianira N. Alfarano for fruitful discussions and comments.
This research was supported through the programme “Research in Pairs” (ID 2145q) by the Mathematisches Forschungsinstitut Oberwolfach in November 2021. The first and last authors are deeply grateful to the Institute for its hospitality.
The research of the last two authors was partially supported by the Italian National Group for Algebraic and Geometric Structures and their Applications (GNSAGA - INdAM).
The research of the last author was supported by the project ``VALERE: VanviteLli pEr la RicErca" of the University of Campania ``Luigi Vanvitelli''. 

\bibliographystyle{abbrv}
\bibliography{biblio}
\end{document}